\documentclass[twocolumn,secnumarabic,amssymb, nobibnotes, aps, prd, superscriptaddress]{revtex4-2}
\usepackage{amsmath,amssymb,amsfonts}
\usepackage{amsthm}
\usepackage{bm}
\usepackage{braket}

\usepackage{graphicx}
\usepackage{epstopdf}
\usepackage{booktabs}

\usepackage{float}
\usepackage[justification=centering]{caption}
\usepackage{subcaption}

\usepackage[ruled,vlined]{algorithm2e}
\usepackage[title]{appendix}

\usepackage{adjustbox}
\usepackage{textcomp}
\usepackage{mathrsfs}

\usepackage[
bookmarks=true,
colorlinks=true,
linkcolor=blue,
urlcolor=blue,
citecolor=blue,
bookmarks=true,
hyperindex=true
]{hyperref}

\makeatletter

\newcommand{\Rmnum}[1]{\expandafter\@slowromancap\romannumeral #1@}
\makeatother

\begin{document}

\title{Local distinguishability of five orthogonal product states on bipartite and tripartite quantum systems}

\author{Guang-Bao \surname{Xu}}
\affiliation{College of Computer Science and Engineering, Shandong University of Science and Technology, Qingdao, 266590, China}

\author{Zi-Yan \surname{Hao}}
\affiliation{College of Computer Science and Engineering, Shandong University of Science and Technology, Qingdao, 266590, China}

\author{Hua-Kun \surname{Wang}}
\affiliation{College of Computer Science and Engineering, Shandong University of Science and Technology, Qingdao, 266590, China}

\author{Yu-Guang \surname{Yang}}
\affiliation{Faculty of Information Technology, Beijing University of Technology, Beijing, 100124, China}

\author{Dong-Huan \surname{Jiang}}
\email{donghuan\_jiang@163.com}
\affiliation{College of Mathematics and Systems Science, Shandong University of Science and Technology, Qingdao, 266590, China}

\date{\today}

\begin{abstract}
	Local distinguishability of orthogonal quantum states can be used to reduce the consumption of quantum resources and lower economic costs in quantum protocols. In this paper, we give the local distinguishability of five orthogonal product states (OPSs) on bipartite and tripartite quantum systems. On one hand, we prove that five OPSs, any two of which are orthogonal only on one subsystem, can be perfectly distinguished by local operations and classical communication (LOCC) except one case, and also discuss the local distinguishability of this exceptional case. On the other hand, we conclude that five tripartite OPSs, any two of which are orthogonal only on one subsystem, can be perfectly distinguished by LOCC except two cases. Furthermore, we discuss the local distinguishability of five tripartite OPSs corresponding to these two exceptional cases. Our work provides a detailed characterization of the local distinguishability of five OPSs.
	\begin{description}
		\item[PACS numbers]
		03.67Mn, 0365.Ud
	\end{description}
\end{abstract}

\pacs{Valid PACS appear here}

\maketitle

\section{\label{sec1}Introduction\protect}
The local distinguishability of orthogonal quantum states, which plays an important role in quantum information theory, obtains a lot of attention since it was proposed. In 1999, Bennett et al. discovered that each of two sets of orthogonal product states (OPSs) cannot be perfectly distinguished by local operations and classical communication (LOCC) \cite{Bennett1999}.  For simplicity, we say a set of OPSs is nonlocal or locally indistinguishable if it cannot be perfectly distinguished by LOCC. In the same year, Bennett et al. gave a method to construct an unextendable product basis (UPB), and demonstrated that a UPB is locally indistinguishable \cite{BennettDiVincenzo1999}. All of these works are of pioneering significance.

Inspired by Bennett et al.'s works, many scholars began to engage in the research on quantum nonlocality and have achieved a lot of research results. Walgate et al. proved that any two orthogonal states can be exactly identified by LOCC \cite{Walgate2000}. Later, Walgate et al. proved that three orthogonal states on $\mathbb{C}^{2}\otimes\mathbb{C}^{2}$ system can be perfectly distinguished if and only if at least two of them are product states, while four orthogonal states on $\mathbb{C}^{2}\otimes\mathbb{C}^{2}$ system can be perfectly distinguished if and only if all of them are product states \cite{Walgate2002}. DiVincenzo et al. pointed out that three or fewer OPSs on multipartite quantum system can be perfectly distinguished by LOCC, while four or fewer bipartite OPSs can also be perfectly distinguished by LOCC \cite{Divincenzo2003}. These results provide profound insights into the problem of local distinguishability on multipartite quantum systems.

During the research process of local distinguishability, scholars adhere to the principles of progressing from bipartite quantum systems to multipartite systems and from simplicity to complexity. First, much progress has been made on the construction methods of nonlocal sets of bipartite OPSs. Zhang et al. gave a method to construct a bipartite nonlocal set of OPSs on bipartite high dimensional quantum systems with equal and odd local dimensions \cite{Zhang2014}. Wang et al. proposed a method to construct a bipartite nonlocal set on a general quantum system \cite{YLWang2015}. Subsequently, some results, which are on the construction methods of bipartite nonlocal state sets \cite{ZCZ2015,GBX201631048,Xu2023,HQ2023}, are made. Second, significant progress has also been achieved in the study of multipartite nonlocal sets of OPSs. Xu et al. proposed a method to construct multipartite nonlocal sets of OPSs on general quantum systems \cite{Xu2016}. Zhang et al. gave a method to construct a nonlocal set of multipartite OPSs by using a nonlocal set of bipartite OPSs \cite{Zhang2017}. Based on existing research foundation, many interesting results on the construction methods of nonlocal sets of multipartite OPSs are proposed \cite{Halder2018,JXu2020,Zhu2022,Zhen2022,Zhen2025}. Third, an important research direction of quantum nonlocality has aroused widespread interest. Halder found that there exist nonlocal sets of OPSs which are locally irreducible in every bipartition, and called this phenomenon strong nonlocality \cite{Halder2019}. Many significant contributions have been made to the research of strong quantum nonlocality \cite{Yuan2020,Hu2020,Wang2021,ShiS2022,ShiYe2022,Zhou2023}. Finally, a novel concept, i.e., minimum nonlocality, is proposed  and some milestone achievements have been attained \cite{Zhu2023,ZhangXu2025}. All the above-mentioned works indicate that the research of quantum nonlocality is meaningful.

Although significant achievements have been made in quantum nonlocality \cite{Wangchen2022,Liz2022,Xiong2023}, there are still many fundamental problems that have not been effectively solved. For example, the local distinguishability of four and five OPSs on multipartite quantum systems remains unsolved. Recently, local distinguishability problem of four OPSs on multipartite quantum systems has been solved in Ref. \cite{XUHWYJ2025}. Compared to the local distinguishability of four OPSs, that of five OPSs is more complex. In this paper, we give local distinguishability of five OPSs on bipartite and tripartite quantum systems, where any two of the five OPSs are orthogonal only on one subsystem. We prove that five bipartite OPSs can be perfectly distinguished by LOCC except one case and five tripartite OPSs can be perfectly distinguished by LOCC except two cases. Furthermore, we discuss the local distinguishability of five bipartite OPSs corresponding to one exceptional case, and the local distinguishability of five tripartite OPSs corresponding to two exceptional cases respectively. The rest of this paper is organized as follows. In Sec.~\ref{sec2}, some preliminaries and basic definitions are given. In Sec.~\ref{sec3}, we discuss the local distinguishability of five bipartite OPSs. In Sec.~\ref{sec4}, we give the local distinguishability of five tripartite OPSs. In Sec.~\ref{sec5}, a brief conclusion is provided. 

\section{\label{sec2}Preliminaries}
\theoremstyle{remark}
\newtheorem{definition}{\indent Definition}
\newtheorem{lemma}{\indent Lemma}
\newtheorem{theorem}{\indent Theorem}
\newtheorem{corollary}{\indent Corollary}
\def\QEDclosed{\mbox{\rule[0pt]{1.3ex}{1.3ex}}}
\def\QED{\QEDclosed}
\def\proof{\indent{\em Proof}.}
\def\endproof{\hspace*{\fill}~\QED\par\endtrivlist\unskip}

In this section, some preliminaries, which will be used in what follows, are introduced. 

\begin{definition}\label{def1}\cite{Bondy2008}
	(Undirected graph) An undirected graph $G$ is a pair $(V,E)$ where $V$ is a finite, non-empty set of vertices and $E$ is a collection of unordered pairs of distinct vertices, called edges. For an edge $e=\{u,v\}\in E$, we say that $e$ joins the vertices $u$ and $v$; equivalently, $u$ and $v$ are incident with $e$.
\end{definition}

\begin{definition}\label{def2}\cite{Bondy2008}
	(Simple graph) A simple graph is an undirected graph with neither loops nor multiple edges.
\end{definition}

\begin{definition}\label{def3}\cite{Bondy2008}
	(Complete graph) A complete graph on $n$ vertices is a simple graph in which every two distinct vertices are adjacent.
\end{definition}

\begin{definition}\label{def4}\cite{Bondy2008}
	(Degree in a simple graph) Let $G=(V,E)$ be an undirected simple graph, and let $v\in V$ be a vertex of $G$. The degree of $v$,  denoted as $deg$($v$), is the number of edges of $G$ incident with $v$.
\end{definition}

To characterize the structure of a set of OPSs, DiVincenzo
et al. proposed the concept of orthogonality graph for a set of OPSs in Ref. \cite{Divincenzo2003}.
\begin{definition}\cite{Divincenzo2003}
	\label{def5}
	Let $\mathcal{H}=\bigotimes_{i=1}^{m} \mathcal{H}_i$ be a $m$-partite Hilbert space with dim $\mathcal{H}_{i} = d_{i}$. Let $S = \{|\psi _{j}\rangle \equiv \bigotimes_{i=1}^{m}|\varphi_{i,j}\rangle\,\,| j = 1,\,2,\,\ldots,\, n\}$ be an orthogonal product basis in $\mathcal{H}$. We represent S as a graph $G = (V,E_{1}\cup E_{2}\cup \ldots \cup E_{m})$ , where the set of edges $E_{i}$ have color $i$. The states $|\psi_{j}\rangle \in$ S are represented as the vertices $V$. There exists an edge $e$ of color $i$ between the vertices $v_{k}$ and $v_{l}$, i.e. $e \in E_{i}$ , when states $|\psi_{k}\rangle$ and $|\psi_{l}\rangle$ are orthogonal on $\mathcal{H}_i$. Since all the states in the product basis are mutually orthogonal, every vertex is connected to all the other vertices by at least one edge of some color. The graph $G$ is called the orthogonality graph of the product basis.
\end{definition}

According to the above definitions, we know that if any two states in a set of OPSs are orthogonal only on one subsystem, then the orthogonality graph of this set is a simple complete undirected graph.

\begin{definition}\label{def6}
	Let $G = (V,E_{1}\cup E_{2}\cup \ldots \cup E_{m})$ be an orthogonality graph of a set of multipartite OPSs, and let $v\in V$ be a vertex of $G$. The degree of $i$-th party of $v$,  denoted as $deg_{i}$($v$), is the number of edges of $G$ incident with vertex $v$ in same color $i$.
\end{definition}

Note that $deg_{i}$($v$) represents the number of product states orthogonal to the one corresponding to vertex $v$ on the  $i$-th subsystem.

To characterize the structure of a set of tripartite OPSs, a concept, the vector of the numbers of pairwise orthogonal relations for tripartite OPSs, was proposed in Ref. \cite{XUHWYJ2025}.
\begin{definition} \cite{XUHWYJ2025}
	\label{def7}
	(The vector of the numbers of pairwise orthogonal relations for tripartite OPSs) A triple ($a$,\,$b$,\,$c$) is used to represent the numbers of pairwise orthogonal relations among OPSs in a set on tripartite quantum system, where $a$ denotes the number of pairwise orthogonal relations among those states on the first subsystem while $b$ and $c$ denote the numbers of pairwise orthogonal relations on the second subsystem and the third subsystem, respectively. For simplicity, we call the triple ($a$,\,$b$,\,$c$) the vector of the numbers of pairwise orthogonal relations. 
\end{definition}

Similarly, we give the concept of the numbers of the vector of pairwise orthogonal relations for bipartite OPSs to characterize the structures of a set of bipartite OPSs.

\begin{definition}
	\label{def8}
	(The vector of the numbers of pairwise orthogonal relations for bipartite OPSs) A vector ($a$,\,$b$) is used to represent the numbers of pairwise orthogonal
relations for bipartite OPSs, where $a$ denotes the number of pairwise orthogonal relations among those states on the first subsystem while $b$ denotes the number of pairwise orthogonal relations on the second subsystem. For simplicity, we call ($a$,\,$b$) the vector of the numbers of pairwise orthogonal relations. 
\end{definition} 

Note that the vector of the numbers of orthogonal relations, ($a$,\,$b$), represents the numbers of edges with two different colors in the orthogonality graph of a set of bipartite OPSs.

Local rank is an important concept to characterize the local distinguishability of a set of OPSs, which was proposed in Ref. \cite{XUHWYJ2025}.  
\begin{definition}\label{def9}\cite{XUHWYJ2025}
	Let $\mathcal{H}=\bigotimes_{i=1}^{m} \mathcal{H}_{i}$ be a $m$-partite Hilbert space with dim $\mathcal{H}_{i} = d_{i}$. Let $S = \{|\psi _{j}\rangle \equiv \bigotimes_{i=1}^{m}|\varphi_{i,j}\rangle\,\,| j = 1,\,2,\,\ldots,\, n\}$ be an orthogonal product basis in $\mathcal{H}$.  We refer to the rank of vector set $\{ |\varphi_{i,j}\rangle\,\,| j = 1,\,2,\,\ldots,\, n\}$ as the local rank $r_{i}$ of $S$ on the $i$-th subsystem for $i=1,\,2,\,\ldots,\,m$. 
\end{definition}

\begin{definition}\cite{Nielsen2010}\label{def10}
	 Suppose a measurement described by measurement operators $\{M_{m}|m=1,\,2,\,\ldots,\,d\}$ is performed upon a quantum system in the state $\vert \psi \rangle$. We define $E_{m} \equiv M^{\dag}_{m}M_{m}$, where $E_{m}$ is a positive operator such that $\sum_{m}E_{m}=I$. The operators $E_{m}$
	are known as the Positive Operator-Valued Measure (POVM) elements associated with the measurement. The complete set
	$\{E_{m}\}$ is known as a POVM. 
\end{definition}

In Ref. \cite{Divincenzo2003}, DiVincenzo et al. gave the local distinguishability of four or fewer bipartite OPSs and the local distinguishability of three multipartite OPSs, which are shown as follows.

 \begin{lemma}\cite{Divincenzo2003}
    \label{lemma1}
     Let $S$ be a set of bipartite OPSs with four or fewer members in any dimension (that allows for this set of OPSs). The set $S$ is distinguishable by local incomplete von Neumann measurements and classical communication.
 \end{lemma}
 \begin{lemma}\cite{Divincenzo2003}
	\label{lemma2}
	A set of any multipartite OPSs with three or fewer members is distinguishable by local incomplete von Neumann measurements and classical communication.
 \end{lemma}

  It should be noted that von Neumann measurement denotes projection measurement, which is a special POVM here.
    
  According to the vectors of the numbers of pairwise orthogonal relations, four tripartite OPSs can be classified into three categories \cite{XUHWYJ2025}, i.e., (4, 1, 1), (3, 2, 1) and (2, 2, 2). The local distinguishability of four tripartite OPSs is shown in the following Lemmas.

  \begin{lemma}\cite{XUHWYJ2025}
	\label{lemma3}
  Four tripartite OPSs with the vector of the numbers of pairwise orthogonal relations $(4,\,1,\,1)$ or $(3,\,2,\,1)$ can be perfectly distinguished by LOCC.
  \end{lemma}
   For category (2, 2, 2), there exist three different orthogonality graphs, i.e., cases (1-1), (1-2) and (1-3), as shown in Fig.~\ref{fig1}.
   \begin{figure}[H]
   	\setlength{\belowcaptionskip}{0.3cm}
   	\centering
   	\includegraphics[width=0.4\textwidth]{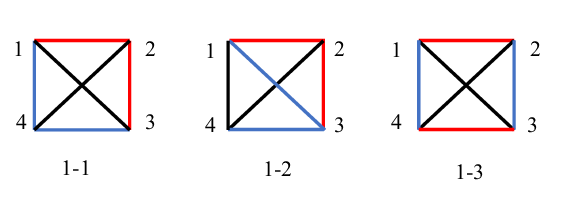}
   	\begin{center}
   		\caption{The feasible graphs of four tripartite OPSs with the vector of the numbers of pairwise orthogonal relations $(2,\,2,\,2)$}\label{fig1}
   		\vspace{-30pt}
   	\end{center}
   \end{figure}
   \begin{lemma}\cite{XUHWYJ2025}
	\label{lemma4}
  	Four tripartite OPSs with orthogonality graph (1-1) or (1-2) can be perfectly distinguished by LOCC when their vector of the numbers of pairwise orthogonal relations is (2, 2, 2).
   \end{lemma}
  \begin{lemma} \label{lemma5} \cite{XUHWYJ2025}
		Four tripartite OPSs with orthogonality graph (1-3) cannot be perfectly distinguished by LOCC when the local ranks of the four tripartite OPSs hold 
    $r_{1}=r_{2}=r_{3}=2$; four tripartite OPSs with orthogonality graph (1-3) can be distinguished by LOCC with some certain probability for any other case.
   \end{lemma}
\section{Local distinguishability of five bipartite orthogonal product states}\label{sec3}

   In this section, we analyze the distinguishability of five bipartite OPSs, where any two states are orthogonal only on one subsystem. Indeed, five bipartite OPSs exactly have ten pairwise orthogonal relations if any two of these states are orthogonal only on one subsystem. Clearly, the local distinguishability of five bipartite OPSs with the vector of the numbers of pairwise orthogonality relations, $(a,b)$,  is invariant under interchange of parties. Thus, we only need to consider one case of the vectors of the numbers of pairwise orthogonality relations $(a,b)$ and $(b,a)$ for five bipartite OPSs.
   
   \begin{figure}[H]
    	\setlength{\belowcaptionskip}{0.3cm}
    	\centering
    	\includegraphics[width=0.3\textwidth]{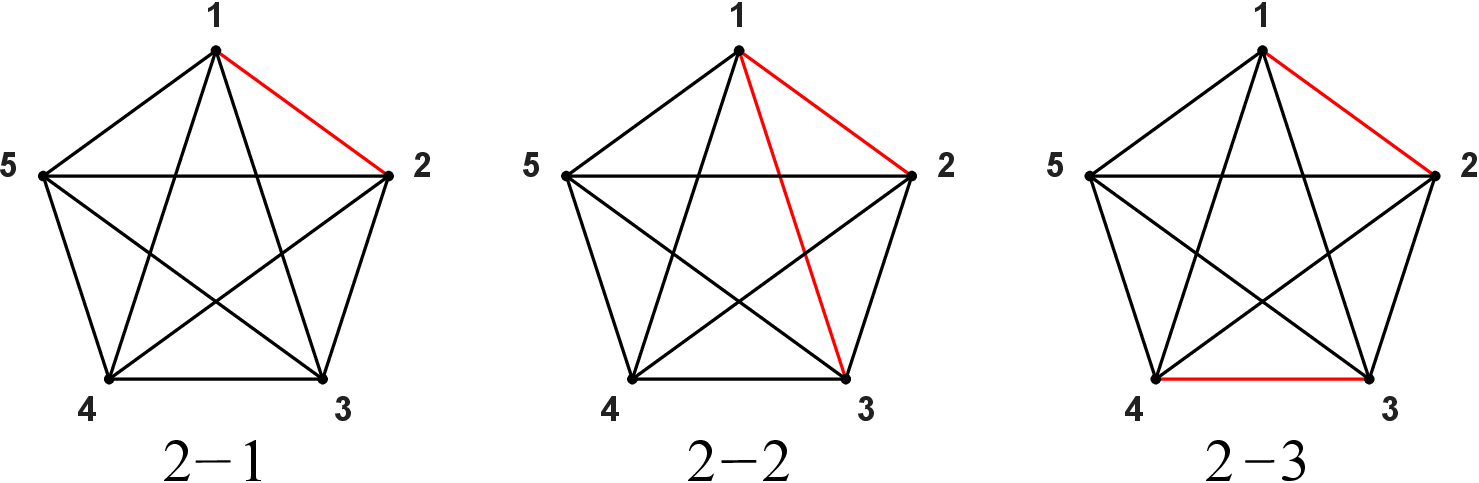}
    	\begin{center}
    		\caption{The feasible graphs of five  bipartite OPSs with the vector of the numbers of pairwise orthogonal relations $(9,\,1)$ and $(8,\,2)$}\label{fig2}
    		\vspace{-30pt}
    	\end{center}
    \end{figure}
    
   \begin{figure}[H]
   	\setlength{\belowcaptionskip}{0.5cm}
   	\centering
   	\includegraphics[width=0.35\textwidth]{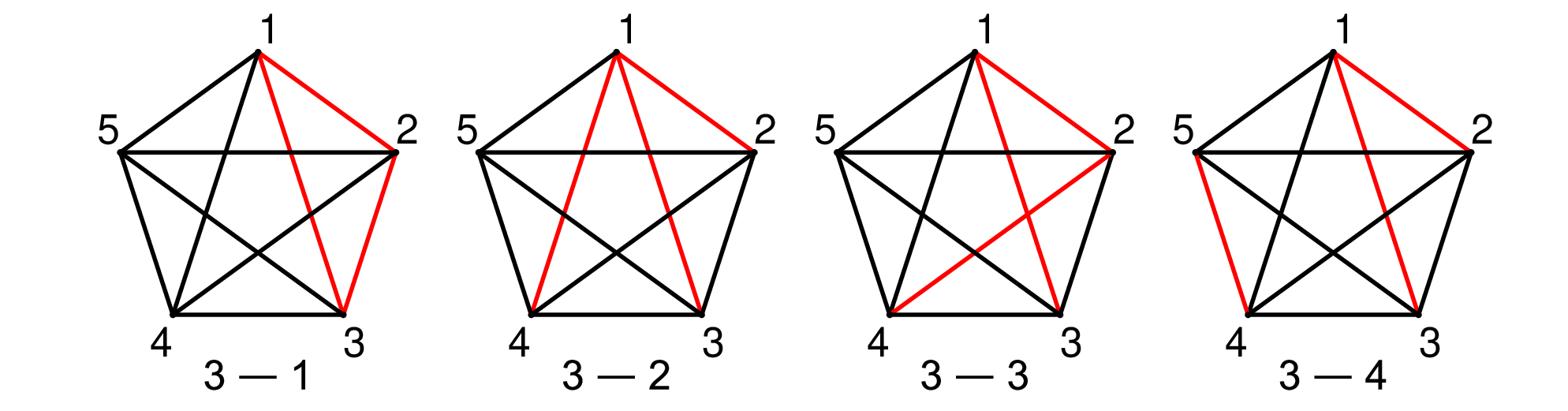}
   	\begin{center}
   		\caption{The feasible graphs of five bipartite OPSs with the vector of the numbers of pairwise orthogonal relations $(7,\,3)$}\label{fig3}
   		\vspace{-30pt}
   	\end{center}
   \end{figure}
  
   \begin{figure}[H]
    	\setlength{\belowcaptionskip}{0.3cm}
    	\centering
    	\includegraphics[width=0.4\textwidth]{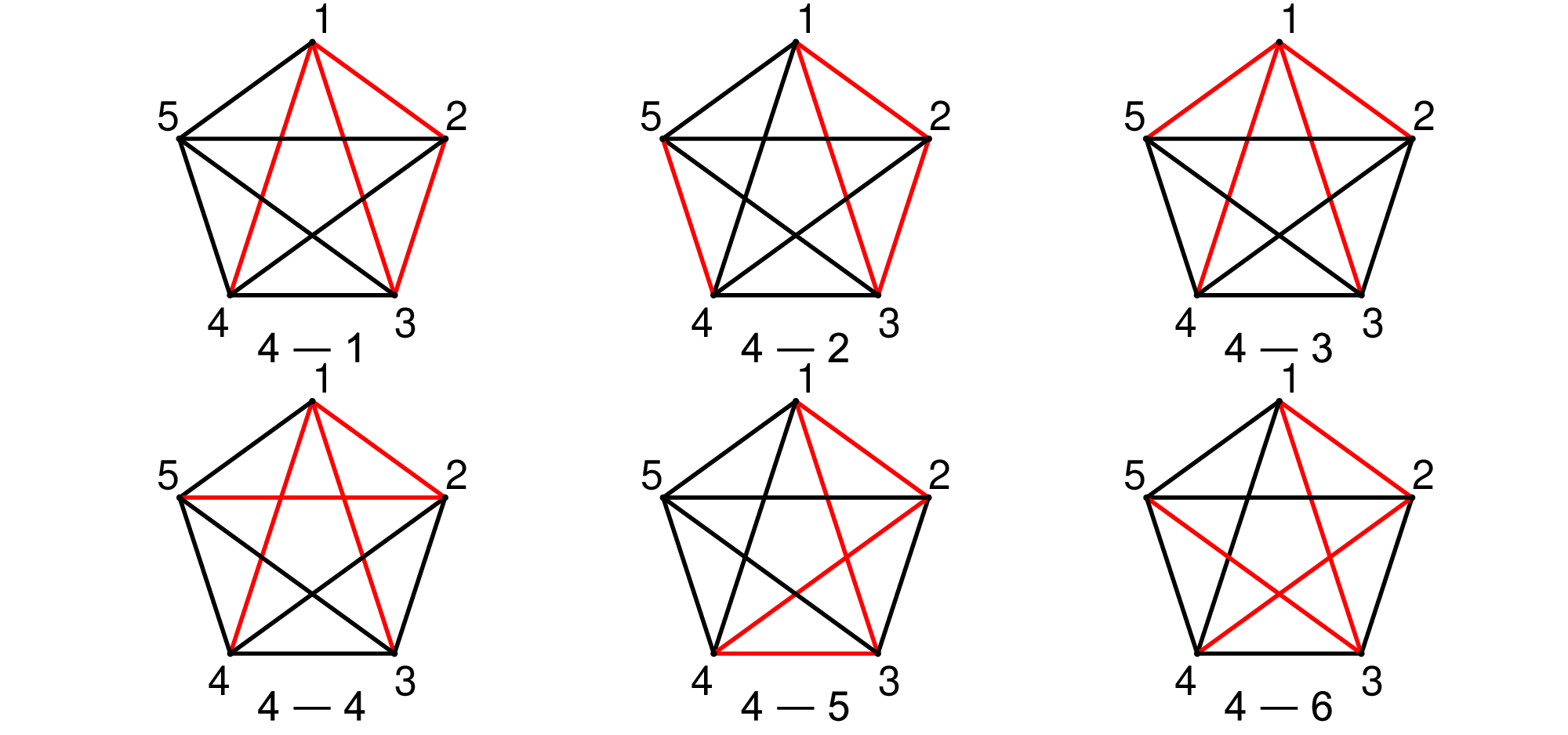}
    	\begin{center}
    		\caption{The feasible graphs of five bipartite OPSs with the vector of the numbers of pairwise orthogonal relations $(6,\,4)$}\label{fig4}
    		\vspace{-30pt}
    	\end{center}
    \end{figure}

    \begin{figure}[H]
     	\setlength{\belowcaptionskip}{0.3cm}
     	\centering
     	\includegraphics[width=0.4\textwidth]{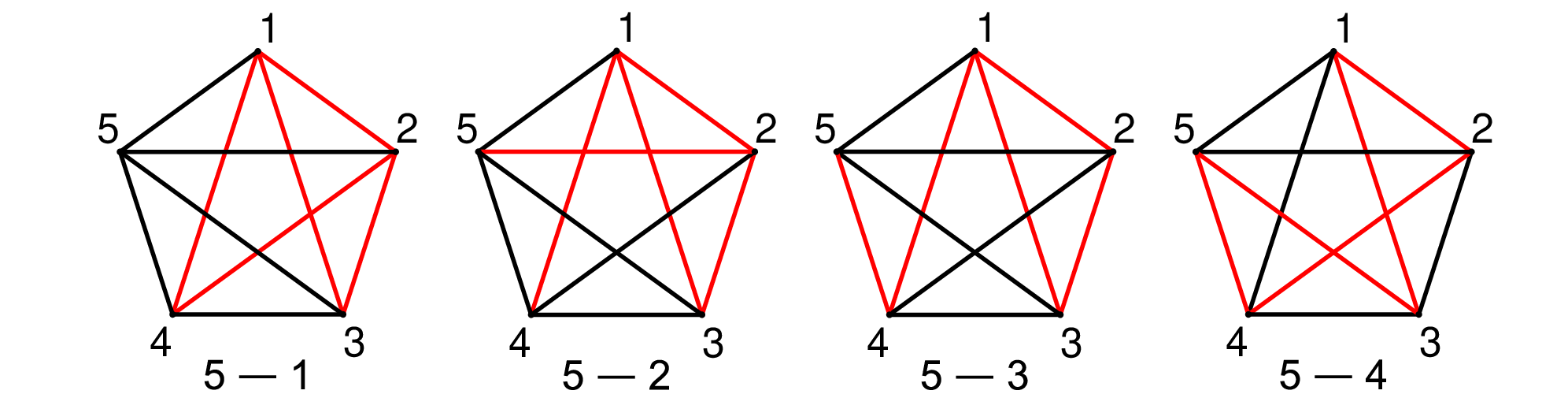}
     	\begin{center}
     		\caption{The feasible graphs of five bipartite OPSs with the vector of the numbers of pairwise orthogonal relations $(5,\,5)$ }\label{fig5}
     		\vspace{-30pt}
     	\end{center}
     \end{figure}
     
      For the structures of five bipartite OPSs, we divide them into six categories according to the vectors of the numbers of pairwise orthogonal relations, i.e.,   $(10,\,0)$, $(9,\,1)$, $(8,\,2)$, $(7,\,3)$, $(6,\,4)$ and $(5,\,5)$. For a set of five bipartite OPSs, the vector of the numbers of pairwise orthogonal relations  $(10,\,0)$ means that all these five OPSs are orthogonal only on one subsystem. Therefore, any five bipartite OPSs with the vector of the numbers of pairwise orthogonal relations  $(10,\,0)$ can be locally distinguished. Now, we consider the local distinguishability of five bipartite OPSs with the vectors of the numbers of pairwise orthogonal relations  $(9,\,1)$, $(8,\,2)$, $(7,\,3)$, $(6,\,4)$ and $(5,\,5)$. All feasible orthogonal graphs for categories $(9,\,1)$, $(8,\,2)$, $(7,\,3)$, $(6,\,4)$ and $(5,\,5)$ are given in Figs. \ref{fig2}-\ref{fig5}. It should be noted that we have omitted graphs which are the same as the graphs shown under interchange of parties as clearly those cases will follow the same line of reasoning.

  \begin{theorem} \label{theorem1}
  	Five bipartite OPSs, any two of which are orthogonal only on one subsystem, can be perfectly distinguished by LOCC except all vertices $v_{j}$ in their orthogonality graph satisfy  $deg_{1}(v_{j})=deg_{2}(v_{j})=2$ for $j=1,\,2,\,\ldots,\,5$.  
    \end{theorem}
    
  In fact, among all feasible orthogonality graphs of five bipartite OPSs, only orthogonality graph (5-4) in Fig. 5 satisfies  $deg_{1}(v_{j})=deg_{2}(v_{j})=2$ for $j=1,\,2,\,\ldots,\,5$. Therefore, five bipartite OPSs corresponding to any other case are locally distinguishable. The proof of Theorem~\ref{theorem1} is given in Appendix \ref{app1}.
   
  Next, we give a detailed discussion for five bipartite OPSs with orthogonality graph (5-4) as shown in Fig. \ref{fig5}.
    
  For case (5-4), states form a cycle of orthogonality on the first subsystem, i.e., state 1 is orthogonal to states 4 and 5; state 2 to states 3 and 5; state 3 to states 2 and 4; and state 4 to states 1 and 3.  Suppose that states  1, 2, 3, 4 and 5 are denoted as $|\phi_{1}\rangle$,  $|\phi_{2}\rangle$, $|\phi_{3}\rangle$, $|\phi_{4}\rangle$ and $|\phi_{5}\rangle$, and the first subsystems of states 1, 2, 3, 4 and 5 are denoted as $|\phi_{1}\rangle_{1}$,  $|\phi_{2}\rangle_{1}$, $|\phi_{3}\rangle_{1}$, $|\phi_{4}\rangle_{1}$ and $|\phi_{5}\rangle_{1}$, respectively. Without loss of generality, we assume that the first subsystems of state 1 and state 5 are $|0\rangle$ and $|1\rangle$, respectively. In fact, if the first subsystems of state 1 and state 5 are not as such, we can transform them into this forms through unitary operations. This rule also applies to the subsequent assumptions. Since the first subsystem of state 2 is orthogonal to that of state 5 and is not orthogonal to that of state 1, suppose that the first subsystem of state 2 is $\frac{1}{\sqrt{1+|a|^{2}}}(|0\rangle+a|2\rangle)$, where $a\neq0$.  
    Since the first subsystem of state 3 is orthogonal to that of state 2 but not to those of states 1 and 5, it must be of the form $\frac{1}{\sqrt{1+|b|^{2}+|c|^{2}+|d|^{2}}}(|0\rangle+b|1\rangle+c|2\rangle+d|3\rangle)$, where $b\neq0$ and $c=-1/a^{*}$. Since the first subsystem of state 4 is orthogonal to those of states $\{1,\,3\}$ but not to those of states $\{2,\,5\}$, it must be of the form $\frac{1}{\sqrt{1+|e|^{2}+|g|^{2}+|h|^{2}}}(|1\rangle+e|2\rangle+g|3\rangle+h|4\rangle)$, where $e=a(b^{*}+d^{*}g)\neq0$. Therefore, the general forms of the first subsystems of states $\{1,\,2,\,3,\,4,\,5\}$ are given as shown in Eq. (1).
      \begin{align}
        |\phi_{1}\rangle_{1}&=|0\rangle,\notag \\
        |\phi_{2}\rangle_{1}&=\frac{1}{\sqrt{1+|a|^{2}}}(|0\rangle+a|2\rangle),\notag\\
        |\phi_{3}\rangle_{1}&=\frac{1}{\sqrt{1+|b|^{2}+|c|^{2}+|d|^{2}}}(|0\rangle+b|1\rangle+c|2\rangle+d|3\rangle),\notag  \\
        |\phi_{4}\rangle_{1}&=\frac{1}{\sqrt{1+|e|^{2}+|g|^{2}+|h|^{2}}}(|1\rangle+e|2\rangle+g|3\rangle+h|4\rangle),\notag \\ 
        |\phi_{5}\rangle_{1}&=|1\rangle,      
    \end{align}
   where $a\neq 0$, $b\neq0$, $c=-1/a^{*}\neq0$, and $e=a(b^{*}+d^{*}g)\neq0$.
    
    For five bipartite OPSs with orthogonal graph (5-4), their second subsystems resemble their first subsystems in structure. Based on the above analysis, we assume that five bipartite OPSs with orthogonal graph (5-4) have the general forms as shown in Eq. (2).
     \begin{align}
        |\phi_{1}\rangle&=\frac{1}{\sqrt{1+|a^{\prime}|^{2}}}|0\rangle_{1}(|0\rangle+a^{\prime}|2\rangle)_{2},\notag \\
        |\phi_{2}\rangle&=\frac{1}{\sqrt{1+|a|^{2}}}(|0\rangle+a|2\rangle)_{1}|1\rangle_{2},\notag
      \end{align}
       \begin{widetext}
        \begin{align}
        |\phi_{3}\rangle&=\frac{1}{\sqrt{1+|b|^{2}+|c|^{2}+|d|^{2}}\sqrt{1+|b^{\prime}|^{2}+|c^{\prime}|^{2}+|d^{\prime}|^{2}}}(|0\rangle+b|1\rangle+c|2\rangle+d|3\rangle)_{1}(|0\rangle+b^{\prime}|1\rangle+c^{\prime}|2\rangle+d^{\prime}|3\rangle)_{2},\notag   \\
        |\phi_{4}\rangle&=\frac{1}{\sqrt{1+|e|^{2}+|g|^{2}+|h|^{2}}}(|1\rangle+e|2\rangle+g|3\rangle+h|4\rangle)_{1}|0\rangle_{2},\notag \\ 
        |\phi_{5}\rangle&=\frac{1}{\sqrt{1+|e^{\prime}|^{2}+|g^{\prime}|^{2}+|h^{\prime}|^{2}}}|1\rangle_{1}(|1\rangle+e^{\prime}|2\rangle+g^{\prime}|3\rangle+h^{\prime}|4\rangle)_{2},     
    \end{align}
     \end{widetext} 
    where  $a\neq 0$, $b\neq0$, $c=-1/a^{*}\neq0$, $e=a(b^{*}+d^{*}g)\neq0$, $a^{\prime}\neq 0$, $b^{\prime}\neq0$, $c^{\prime}=-1/(a^{\prime})^{*}\neq0$, and $e^{\prime}=a^{\prime}[(b^{\prime})^{*}+(d^{\prime})^{*}g^{\prime}]\neq0$.   
            
    \begin{theorem} \label{theorem2}
	Five bipartite OPSs with orthogonality graph (5-4), as shown in Eq. (2), cannot be perfectly distinguished by LOCC when 
    $h=d=g=h^{\prime}=d^{\prime}=g^{\prime}=0$.    
    \end{theorem}
    \begin{proof}
    When $h=d=g=h^{\prime}=d^{\prime}=g^{\prime}=0$, the five states in Eq. (2) become the forms as shown in Eq. (3).
    
    \begin{align}
        |\phi_{1}\rangle&=\frac{1}{\sqrt{1+|a^{\prime}|^{2}}}|0\rangle_{1}(|0\rangle+a^{\prime}|2\rangle)_{2},\notag\\
        |\phi_{2}\rangle&=\frac{1}{\sqrt{1+|a|^{2}}}(|0\rangle+a|2\rangle)_{1}|1\rangle_{2},\notag\\
        |\phi_{3}\rangle&=\frac{1}{\sqrt{1+|b|^{2}+|c|^{2}}\sqrt{1+|b^{\prime}|^{2}+|c^{\prime}|^{2}}}\times\notag\\
            &(|0\rangle+b|1\rangle+c|2\rangle)_{1}(|0\rangle+b^{\prime}|1\rangle+c^{\prime}|2\rangle)_{2},\notag\\
        |\phi_{4}\rangle&=\frac{1}{\sqrt{1+|e|^{2}}}(|1\rangle+e|2\rangle)_{1}|0\rangle_{2},\notag 
        \end{align}
        \begin{align}
        |\phi_{5}\rangle&=\frac{1}{\sqrt{1+|e^{\prime}|^{2}}}|1\rangle_{1}(|1\rangle+e^{\prime}|2\rangle)_{2}.     
    \end{align}
        
    Without loss of generality, suppose that the first party, say Alice,  first performs an orthogonality-preserving measurement with POVM elements $M_{j}^{\dagger}M_{j}=$ 
    $$\begin{pmatrix}
		m_{00}^{(j)} & m_{01}^{(j)} & m_{02}^{(j)} &m_{03}^{(j)} &\cdots &m_{0,n-1}^{(j)}\\
		m_{10}^{(j)} & m_{11}^{(j)} & m_{12}^{(j)} &m_{13}^{(j)} &\cdots &m_{1,n-1}^{(j)}\\
        m_{20}^{(j)} & m_{21}^{(j)} & m_{22}^{(j)} &m_{23}^{(j)} &\cdots &m_{2,n-1}^{(j)}\\
        m_{30}^{(j)} & m_{31}^{(j)} & m_{32}^{(j)} &m_{33}^{(j)} &\cdots &m_{3,n-1}^{(j)}\\
        \vdots & \vdots & \vdots & \vdots &\ddots &\vdots\\
		m_{n-1,0}^{(j)} & m_{n-1,1}^{(j)} & m_{n-1,2}^{(j)} &m_{n-1,3}^{(j)} &\cdots &m_{n-1,n-1}^{(j)} \\ 
	\end{pmatrix}$$\\
    under the basis $\{|0\rangle,\,|1\rangle,\cdots,\,|(n-1)\rangle\}$ for $j=1$, $2$, $\cdots$, $l$, where $n\geq3$, and $n$ denotes the dimension of the space in which the first subsystems of these five states reside. To ensure that the measurement can proceed, any two states that are orthogonal only on Alice's side should remain orthogonal after being measured by Alice.      
     
    For the OPSs $|\phi_{1}\rangle$ and $|\phi_{5}\rangle$, we have $\langle\phi_{1}|M_{j}^{\dagger}M_{j}$ $\otimes I_{B}|\phi_{5}\rangle=0$ and $\langle \phi_{5}|M_{j}^{\dagger}M_{j}\otimes I_{B}|\phi_{1}\rangle=0$. Thus, $m_{01}^{(j)}=0$ and $m_{10}^{(j)}=0$. For $|\phi_{1}\rangle$ and $|\phi_{4}\rangle$, we have  $\langle \phi_{1}|M_{j}^{\dagger}M_{j}\otimes I_{B}|\phi_{4}\rangle=0$ and $\langle \phi_{4}|M_{j}^{\dagger}M_{j}\otimes I_{B}|\phi_{1}\rangle=0$. That is, $m_{01}^{(j)}+em_{02}^{(j)}=0$ and $m_{10}^{(j)}+e^{*}m_{20}^{(j)}=0$. Thus, we have $m_{02}^{(j)}=0$ and $m_{20}^{(j)}=0$ since $m_{01}^{(j)}=0$, $m_{10}^{(j)}=0$ and $e\neq0$. For $|\phi_{2}\rangle$ and $|\phi_{5}\rangle$, we have  $\langle \phi_{2}|M_{j}^{\dagger}M_{j}\otimes I_{B}|\phi_{5}\rangle=0$ and  $\langle \phi_{5}|M_{j}^{\dagger}M_{j}\otimes I_{B}|\phi_{2}\rangle=0$. That is, $m_{01}^{(j)}+a^{*}m_{21}^{(j)}=0$ and $m_{10}^{(j)}+am_{12}^{(j)}=0$. Thus, we have $m_{21}^{(j)}=0$ and $m_{12}^{(j)}=0$ since $m_{01}^{(j)}=m_{10}^{(j)}=0$ and $a\neq0$. For $|\phi_{2}\rangle$ and $|\phi_{3}\rangle$, we have  $\langle \phi_{2}|M_{j}^{\dagger}M_{j}\otimes I_{B}|\phi_{3}\rangle=0$. Thus, $m_{00}^{(j)}+a^{*}cm_{22}^{(j)}=0$. Since $c=-1/a^{*}$, we have $m_{00}^{(j)}=m_{22}^{(j)}$. For $|\phi_{4}\rangle$ and $|\phi_{3}\rangle$, we have  $\langle \phi_{4}|M^{\dagger}_{j}M_{j}\otimes I_{B}|\phi_{3}\rangle=0$. Thus, $bm_{11}^{(j)}+e^{*}cm_{22}^{(j)}=0$. Since $e=ab^{*}$ and $c=-1/a^{*}$, we have $m_{11}^{(j)}=m_{22}^{(j)}$. Consequently, each POVM element $M_{j}^{\dagger}M_{j}$ must have the form\\
    $$\begin{pmatrix}
		 m_{00}^{(j)} & 0 & 0 &m_{03}^{(j)} &\cdots &m_{0,n-1}^{(j)}\\
		 0 & m_{00}^{(j)} & 0 &m_{13}^{(j)} &\cdots &m_{1,n-1}^{(j)}\\
         0 & 0 & m_{00}^{(j)} &m_{23}^{(j)} &\cdots &m_{2,n-1}^{(j)}\\
         m_{30}^{(j)} & m_{31}^{(j)} & m_{32}^{(j)} &m_{33}^{(j)} &\cdots &m_{3,n-1}^{(j)}\\
         \vdots & \vdots & \vdots & \vdots &\ddots &\vdots\\
		 m_{n-1,0}^{(j)} & m_{n-1,1}^{(j)} & m_{n-1,2}^{(j)} &m_{n-1,3}^{(j)} &\cdots &m_{n-1,n-1}^{(j)} \\ 
	\end{pmatrix}.$$\\
    Thus, when Alice's measurement outcome is $j$, the probability that the outcome $j$ occurs is $p(j)=\langle\phi_{k}|M_{j}^{\dag}M_{j}\otimes I_{B}|\phi_{k}\rangle=m_{00}^{(j)}$ for $k\in\{1,\,2,\,3,\,4,\,5\}$, where $j\in\{1,\,$ $2,\,$ $\cdots,\,$ $l\}$. This means that the probability of measurement outcome $j$ occurring is identical for any state $|\phi_{k}\rangle$. Therefore, Alice cannot get any useful information to distinguish these five OPSs. 
    
       In fact, the second party, say Bob, will face the same situation as Alice does due to the symmetry of the set composed of these five OPSs. Therefore, these five states cannot be perfectly distinguished by LOCC. This completes the proof. 
    \end{proof}
    
    For the first subsystems of five OPSs in Eq. (2), there exist three different cases when $h=0$, i.e., (1) $d\neq0$ and $g=0$; (2) $d=0$ and $g\neq0$; (3) $d\neq0$ and $g\neq0$. For the distinguishability of five OPSs in Eq. (2), we have the following conclusion as shown in Theorem~\ref{theorem3}.

    \begin{theorem}\label{theorem3} Suppose that the state to be identified is one of five bipartite OPSs in Eq. (2) with equal likelihood. (1) When $h=g=0$ and $d\neq0$, there exists a protocol that allows the first party to perfectly distinguish the five OPSs in Eq. (2) with the probability of $\frac{|d|^{2}}{5(1+|b|^{2}+|c|^{2}+|d|^{2})}$; (2) when $h=d=0$ and $g\neq0$, there exists a protocol that allows the first party to perfectly distinguish the five OPSs in Eq. (2) with the probability of $\frac{|g|^{2}}{5(1+|e|^{2}+|g|^{2})}$; (3) when $h=0$,  $d\neq0$ and $g\neq0$, there exists a protocol that allows the first party to perfectly distinguish the five OPSs in Eq. (2) with the probability of $\frac{1}{5}\{\frac{|d|^{2}}{(|d|^{2}+|b|^{2})}+\frac{|d-gb|^{2}}{(|d|^{2}+|b|^{2})[1+|a(b^{*}+d^{*}g)|^{2}+|g|^{2}]}\}.$\end{theorem}
    
    \begin{proof} We provide proofs for different cases separately.

     (1) When $h=g=0$ and $d\neq0$ in Eq. (2)
     
     In this case, the five subsystems that Alice needs to distinguish are given as shown in Eq. (4).
     
     \begin{align}
        |\phi_{1}\rangle_{1}&=|0\rangle,\notag \\
        |\phi_{2}\rangle_{1}&=\frac{1}{\sqrt{1+|a|^{2}}}(|0\rangle+a|2\rangle),\notag\\
        |\phi_{3}\rangle_{1}&=\frac{1}{\sqrt{1+|b|^{2}+|c|^{2}+|d|^{2}}}(|0\rangle+b|1\rangle+c|2\rangle+d|3\rangle),\notag  \\
        |\phi_{4}\rangle_{1}&=\frac{1}{\sqrt{1+|e|^{2}}}(|1\rangle+e|2\rangle),\notag \\ 
        |\phi_{5}\rangle_{1}&=|1\rangle,      
    \end{align}
    where $a\neq 0$, $b\neq0$, $c=-1/a^{*}\neq0$,  $e\neq0$ and $d\neq0$. Suppose that the first party, say Alice, performs a measurement with the operators $M_{1}=|3\rangle\langle3|$ and $M_{2}=I-|3\rangle\langle3|$.  
    
    \textcircled{1} If Alice's measurement outcome corresponds to  $M_{1}$, the measured state must be state 3, i.e., $|\phi_{3}\rangle$. The probability of this event is $P(1)=\langle\phi_{3}|M_{1}^{\dag}M_{1}\otimes I_{B}|\phi_{3}\rangle=\langle\phi_{3}|M_{1}^{\dag}M_{1}|\phi_{3}\rangle_{1}=\frac{|d|^{2}}{(1+|b|^{2}+|c|^{2}+|d|^{2})}$ when the measured state is $|\phi_{3}\rangle$. Given that the probability of the state under measurement being state 3 is 1/5, the probability for its perfect identification by Alice is $\frac{|d|^{2}}{5(1+|b|^{2}+|c|^{2}+|d|^{2})}$. 
    
    \textcircled{2} If Alice's measurement outcome corresponds to  $M_{2}$, the first subsystem of the measured state must be one of the forms as shown in Eq. (5) after Alice's measurement.
    \begin{align}
        |\phi_{1}\rangle_{1}&=|0\rangle,\notag\\ 
        |\phi_{2}\rangle_{1}&=\frac{1}{\sqrt{1+|a|^{2}}}(|0\rangle+a|2\rangle),\notag\\
        |\phi_{3}^{\prime}\rangle_{1}&=\frac{1}{\sqrt{1+|b|^{2}+|c|^{2}}}(|0\rangle+b|1\rangle+c|2\rangle),\notag\\
        |\phi_{4}\rangle_{1}&=\frac{1}{\sqrt{1+|e|^{2}}}(|1\rangle+e|2\rangle),\notag \\ 
        |\phi_{5}\rangle_{1}&=|1\rangle.       
    \end{align}
    Note that the states in Eq. (5) are identical to the first subsystems of the states in Eq. (3). This means that Alice faces the same situation as she does in the proof of Theorem~\ref{theorem2}. Therefore, Alice cannot get any useful information to identify the measured state in this situation.
    
    (2) when $h=d=0$ and $g\neq0$ in Eq. (2)
    
    In this case, the five subsystems that Alice needs to distinguish are given as shown in Eq. (6).
        \begin{align}
        |\phi_{1}\rangle_{1}&=|0\rangle,\notag \\
        |\phi_{2}\rangle_{1}&=\frac{1}{\sqrt{1+|a|^{2}}}(|0\rangle+a|2\rangle),\notag\\
        |\phi_{3}\rangle_{1}&=\frac{1}{\sqrt{1+|b|^{2}+|c|^{2}}}(|0\rangle+b|1\rangle+c|2\rangle),\notag  \\
        |\phi_{4}\rangle_{1}&=\frac{1}{\sqrt{1+|e|^{2}+|g|^{2}}}(|1\rangle+e|2\rangle+g|3\rangle),\notag \\ 
        |\phi_{5}\rangle_{1}&=|1\rangle,      
    \end{align}
   where $a\neq 0$, $b\neq0$, $c=-1/a^{*}\neq0$, $e=ab^{*}$. Suppose that the first party, say Alice, performs a measurement with the operators $M_{1}=|3\rangle\langle3|$ and $M_{2}=I-|3\rangle\langle3|$.  
   
   \textcircled{1} If Alice's measurement outcome corresponds to  $M_{1}$, the measured state must be state 4, i.e., $|\phi_{4}\rangle$. The probability of this event is $P(1)=\langle\phi_{4}|M_{1}^{\dag}M_{1}\otimes I_{B}|\phi_{4}\rangle=\langle\phi_{4}|M_{1}^{\dag}M_{1}|\phi_{4}\rangle_{1}=\frac{|g|^{2}}{(1+|e|^{2}+|g|^{2})}$ when the measured state is $|\phi_{4}\rangle$. Given that the probability of the state under measurement being state 4 is 1/5, the probability for its perfect identification by Alice is $\frac{|g|^{2}}{5(1+|e|^{2}+|g|^{2})}$.

   \textcircled{2} If Alice's measurement outcome corresponds to  $M_{2}$, the first subsystem of the measured state must be one of the forms as shown in Eq. (7) after Alice's measurement.
    \begin{align}
        |\phi_{1}\rangle_{1}&=|0\rangle,\notag \\
        |\phi_{2}\rangle_{1}&=\frac{1}{\sqrt{1+|a|^{2}}}(|0\rangle+a|2\rangle),\notag\\
        |\phi_{3}\rangle_{1}&=\frac{1}{\sqrt{1+|b|^{2}+|c|^{2}}}(|0\rangle+b|1\rangle+c|2\rangle),\notag\\
        |\phi_{4}^{\prime}\rangle_{1}&=\frac{1}{\sqrt{1+|e|^{2}}}(|1\rangle+e|2\rangle),\notag \\ 
        |\phi_{5}\rangle_{1}&=|1\rangle.        
    \end{align}
    Note that the states in Eq. (7) are identical to the first subsystems of the states in Eq. (3). This means that Alice faces the same situation as she does in the proof of Theorem~\ref{theorem2}. Therefore, Alice cannot get any useful information to identify the measured state in this situation.

    (3) When $h=0$,  $d\neq0$ and $g\neq0$ in Eq. (2)
    
    In this case, the five subsystems that Alice needs to identify must be one of the forms as shown in Eq. (8).

    \begin{align}
        |\phi_{1}\rangle_{1}&=|0\rangle,\notag \\
        |\phi_{2}\rangle_{1}&=\frac{1}{\sqrt{1+|a|^{2}}}(|0\rangle+a|2\rangle),\notag\\
        |\phi_{3}\rangle_{1}&=\frac{1}{\sqrt{1+|b|^{2}+|\frac{1}{a^{*}}|^{2}+|d|^{2}}}[|0\rangle+b|1\rangle-\frac{1}{a^{*}}|2\rangle+d|3\rangle], \notag \\
        |\phi_{4}\rangle_{1}&=\frac{1}{\sqrt{1+|a(b^{*}+d^{*}g)|^{2}+|g|^{2}}}[|1\rangle+\notag \\ 
         &\qquad  a(b^{*}+d^{*}g)|2\rangle+g|3\rangle],\notag \\ 
        |\phi_{5}\rangle_{1}&=|1\rangle,       
    \end{align}
    where $a\neq 0$, $b\neq0$, $d\neq0$ and $g\neq0$. Suppose that Alice performs a measurement with the operators $M_{1}=\frac{1}{|d|^{2}+|b|^{2}}(d^{*}|1\rangle-b^{*}|3\rangle)(d\langle1|-b\langle3|)$ and $M_{2}= I-\frac{1}{|d|^{2}+|b|^{2}}(d^{*}|1\rangle-b^{*}|3\rangle)(d\langle1|-b\langle3|)$. 
    
    \textcircled{1} If Alice's measurement outcome corresponds to  $M_{1}$, the measured state must be state 4 or 5, i.e., $|\phi_{4}\rangle$ or $|\phi_{5}\rangle$. States 4 and 5 can be perfectly distinguished by the second party since they are orthogonal on the second subsystem. If the state to be identified is $|\phi_{4}\rangle$, the probability of this outcome occurring is 
    $\langle\phi_{4}|M_{1}^{\dag}M_{1}\otimes I_{B}|\phi_{4}\rangle=\langle\phi_{4}|M_{1}^{\dag}M_{1}|\phi_{4}\rangle_{1}=\frac{|d-gb|^{2}}{(|d|^{2}+|b|^{2})[1+|a(b^{*}+d^{*}g)|^{2}+|g|^{2}]}$; 
    If the state to be identified is $|\phi_{5}\rangle$, the probability of this outcome occurring is 
    $\langle\phi_{5}|M_{1}^{\dag}M_{1}\otimes I_{B}|\phi_{5}\rangle=\langle\phi_{5}|M_{1}^{\dag}M_{1}|\phi_{5}\rangle_{1}=\frac{|d|^{2}}{|d|^{2}+|b|^{2}}$. Given that the probability of the state under measurement being state 4 or 5 is 1/5, the probability of the outcome 1 occurring is  $\frac{1}{5}\{\frac{|d|^{2}}{(|d|^{2}+|b|^{2})}+\frac{|d-gb|^{2}}{(|d|^{2}+|b|^{2})[1+|a(b^{*}+d^{*}g)|^{2}+|g|^{2}]}\}.$

    \textcircled{2} If Alice's measurement outcome corresponds to  $M_{2}$, the first subsystem of the measured state must be one of the following forms as shown in Eq. (9) after Alice's measurement.

    \begin{align}
        |\phi_{1}\rangle_{1}&=|0\rangle,\notag \\
        |\phi_{2}\rangle_{1}&=\frac{1}{\sqrt{1+|a|^{2}}}(|0\rangle+a|2\rangle),\notag
    \end{align}
    \begin{widetext}
    \begin{align}      
        |\phi_{3}\rangle_{1}&=\frac{1}{\sqrt{1+|b|^{2}+|\frac{1}{a^{*}}|^{2}+|d|^{2}}}[|0\rangle+b|1\rangle-\frac{1}{a^{*}}|2\rangle+d|3\rangle], \notag\\ 
        |\phi_{4}^{\prime}\rangle_{1}&=\frac{1}{\sqrt{(|b|^{2}+b^{*}g^{*}d)(|b|^{2}+bgd^{*})+(g^{*}|d|^{2}+d^{*}b)(g|d|^{2}+db^{*})}}
        [(|b|^{2}+bgd^{*})|1\rangle+(g|d|^{2}+db^{*})|3\rangle],\notag \\ 
        |\phi_{5}^{\prime}\rangle_{1}&=\frac{1}{\sqrt{|b|^{4}+|d|^{2}|b|^{2}}}[|b|^{2}|1\rangle+db^{*}|3\rangle].       
    \end{align}
     \end{widetext} 
   At this moment, if $|\phi_{3}\rangle_{1}$ and $|\phi_{4}^{\prime}\rangle_{1}$ is orthogonal, i.e., $(|b|^{2}+|d|^{2})(b^{*}+gd^{*})=0$, we have $(b^{*}+gd^{*})=0$. This means  $\langle\phi_{2}|\phi_{4}\rangle_{1}=0$, which contradicts the fact that states 2 and 4 are orthogonal only on the second side. Thus Alice's measurement cannot proceed since $|\phi_{3}\rangle_{1}$ and $|\phi_{4}^{\prime}\rangle_{1}$ is not orthogonal. This completes the proof.
   \end{proof}	
    
   It should be emphasized that the second party will face a similar situation as Alice does in the proof of Theorem~\ref{theorem3} if he performs an orthogonality-preserving measurement on the states chosen from five OPSs in Eq. (2). 
   
   For states in Eq. (2), we give a discrimination protocol from the first side when $h\neq0$. Alice performs a measurement with the operators $M_{1}=|4\rangle\langle4|$ and $M_{2}=I-|4\rangle\langle4|$.  
   
   \textcircled{1} If the measurement outcome corresponds to $M_{1}$, the state to be identified must be state 4, i.e., $|\phi_{4}\rangle$. 
   
   \textcircled{2} If the measurement outcome corresponds to $M_{2}$, the state to be identified is one of states \{1, 2, 3, 4, 5\}. If the state to be identified is one of \{1, 2, 3, 5\}, its first subsystem remains invariant; If the state to be identified is state 4, i.e., $|\phi_{4}\rangle$, its first subsystem collapses to
    $|\phi_{4}^{\prime}\rangle_{1}=\frac{1}{\sqrt{\langle\phi_{4}|M_{2}^{\dag}M_{2}|\phi_{4}\rangle_{1}}}M_{2}|\phi_{4}\rangle_{1}=\frac{1}{\sqrt{1+|e|^{2}+|g|^{2}}}(|1\rangle+e|2\rangle+g|3\rangle).$
    Thus, the first subsystem that Alice needs to identify becomes one of the states $\{|\phi_{1}\rangle_{1},\,|\phi_{2}\rangle_{1},\,|\phi_{3}\rangle_{1},\,|\phi_{4}^{\prime}\rangle_{1},\,|\phi_{5}\rangle_{1}\}$. The subsequent discrimination method can refer to the proof of Theorem~\ref{theorem3}.
    	
\section{Local distinguishability of five orthogonal product states on tripartite quantum systems}\label{sec4}  

   In this section, we further analyze the distinguishability of five tripartite OPSs where any two states are orthogonal only on one subsystem. We still use orthogonal graphs to represent the structures of the sets of five tripartite OPSs. We have excluded graphs that are identical to those shown when parties are interchanged, as those cases will logically follow the same reasoning. In fact, five tripartite OPSs with the vector of the numbers of pairwise orthogonal relations, ($a$, $b$, $c$), ($c$, $a$, $b$), ($b$, $c$, $a$), ($a$, $c$, $b$), ($b$, $a$, $c$) or ($c$, $b$, $a$) exhibit the same local distinguishability. Therefore, it suffices to discuss just one of these six scenarios.
   
   Since the number of pairwise orthogonal relations on each subsystem is not certain, we enumerate all feasible categories as shown in Table~\ref{tab1}. Note that five tripartite OPSs with the vectors of pairwise orthogonal relations, namely, $(10,\,0,\,0)$, $(9,\,1,\,0)$, $(8,\,2,\,0)$, $(7,\,3,\,0)$, $(6,\, 4,\,0)$ and $(5,\,5,\,0)$ can be seen as bipartite OPSs with the vectors of the numbers of pairwise orthogonal relations, $(10,\,0)$, $(9,\,1)$, $(8,\,2)$, $(7,\,3)$, $(6,\,4)$ and $(5,\,5)$, respectively. This means that the local distinguishability of these categories can be reduced to that of five bipartite OPSs. Therefore, we only need to consider the categories of five tripartite OPSs where each party has at least one pair of orthogonal relations.
    \begin{table}[H]
	\centering
	\caption{Categories of five tripartite OPSs by the vectors of the numbers of pairwise orthogonal relations}\label{tab1}
	\begin{ruledtabular}
		\begin{tabular}{c c}
			Categories that can be seen   &Categories that cannot be \\
			as bipartite set of OPSs  &seen as bipartite set of OPSs \\
			\noalign{\smallskip}\hline\noalign{\smallskip}
			(10,\,0,\,0)     & (8,\,1,\,1) \\
			(9,\,1,\,0)     & (7,\,2,\,1) \\
			(8,\,2,\,0)     & (6,\,3,\,1) \\
			(7,\,3,\,0)     & (6,\,2,\,2) \\
			(6,\,4,\,0)     & (5,\,4,\,1) \\
			(5,\,5,\,0)     & (5,\,3,\,2) \\
			                & (4,\,4,\,2) \\
                            & (4,\,3,\,3) \\
		\end{tabular}
	\end{ruledtabular}
\end{table}
  
\begin{figure}[H]
		\setlength{\belowcaptionskip}{0.2cm}
		\centering
		\includegraphics[width=0.25\textwidth]{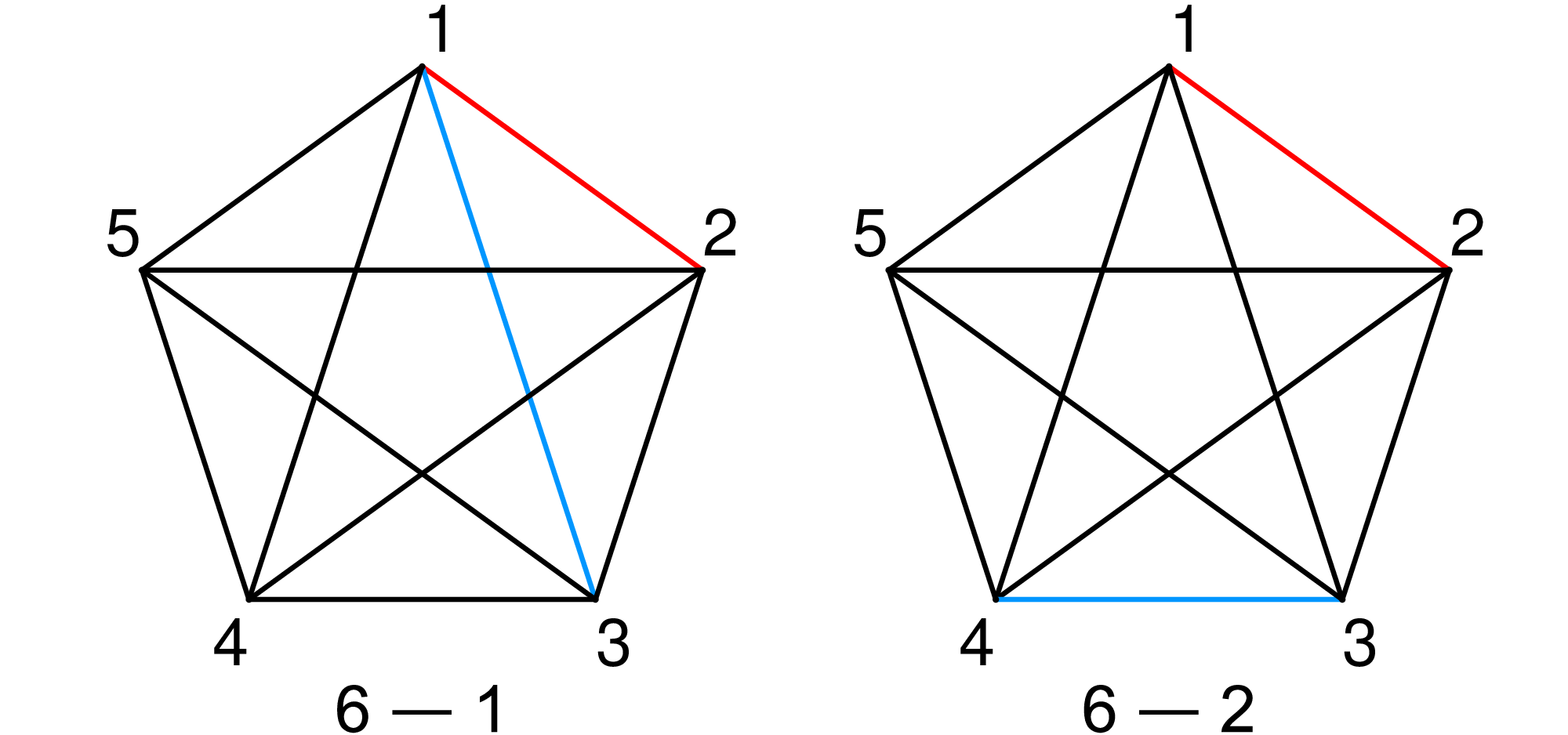}
		\begin{center}
			\caption{The feasible graphs of five OPSs with the vector of the numbers of pairwise orthogonal relations $(8,\,1,\,1)$}\label{fig6}
		\end{center}
	\end{figure}

All feasible orthogonal graphs for categories $(8,\,1,\,1)$, $(7,\,2,\,1)$,  $(6,\,3,\,1)$, $(6,\,2,\,2)$, $(5,\,4,\,1)$, $(5,\,3,\,2)$, $(4,\,4,\,2)$ and $(4,\,3,\,3)$ are given in Figs. \ref{fig6}-\ref{fig13}. It should be noted that we have omitted graphs which are the same as the graphs shown under interchange of parties as clearly those cases will follow the same line of reasoning.

\begin{figure}[H]
		\setlength{\belowcaptionskip}{0.2cm}
		\centering
		\includegraphics[width=0.4\textwidth]{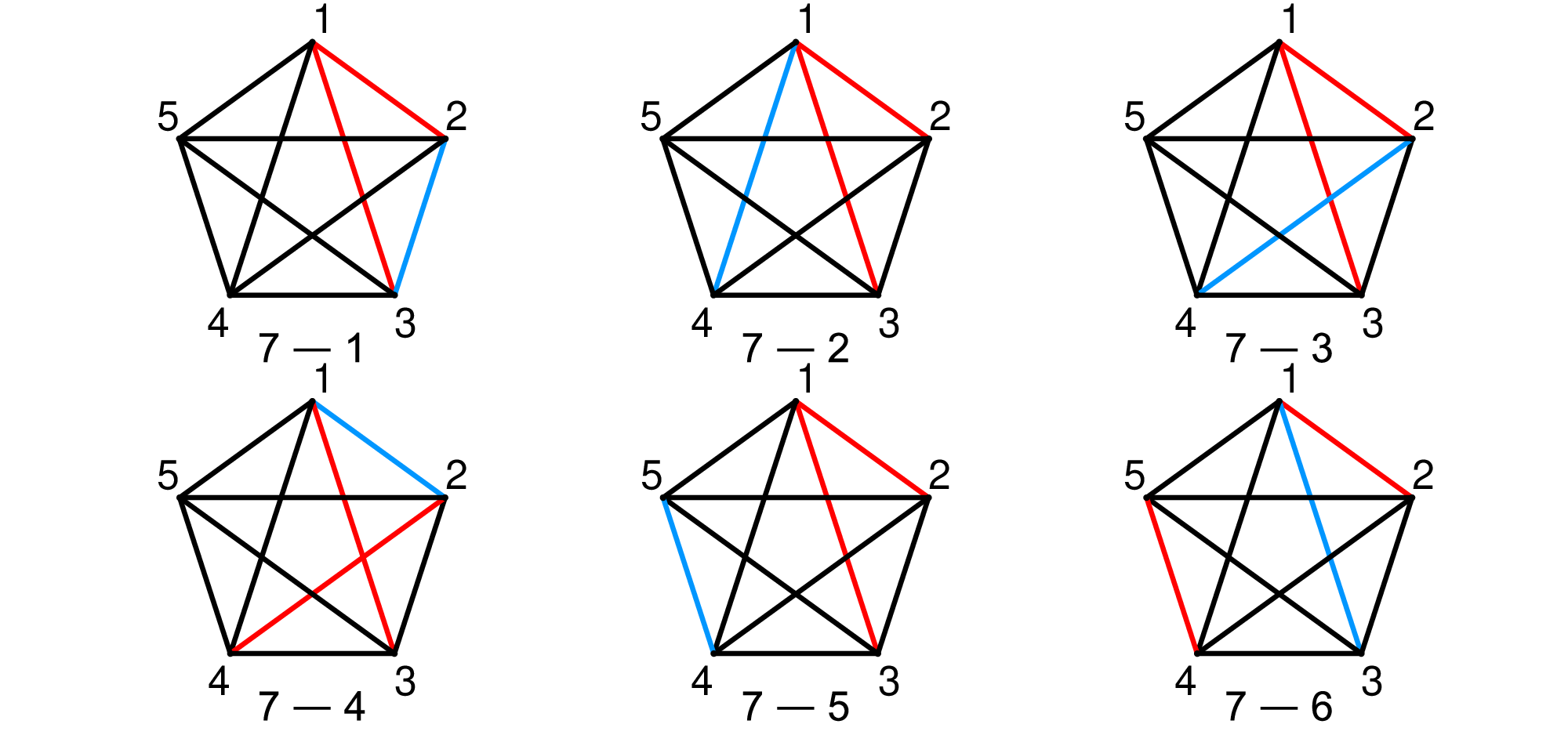}
		\begin{center}
			\caption{The feasible graphs of five OPSs with the vector of the numbers of pairwise orthogonal relations $(7,\,2,\,1)$}\label{fig7}
			\vspace{-40pt}
		\end{center}
	\end{figure}

\begin{figure}[H]
		\setlength{\belowcaptionskip}{0.2cm}
		\centering
		\includegraphics[width=0.40\textwidth]{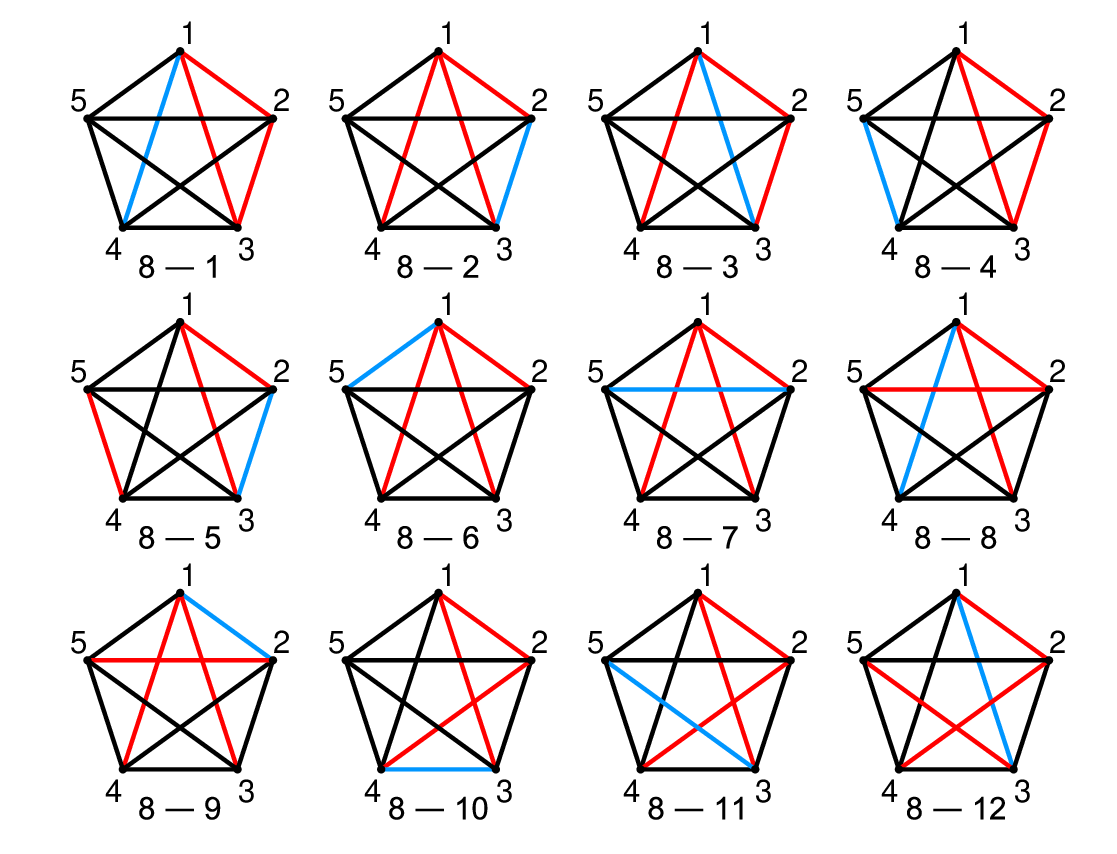}
		\begin{center}
			\caption{The feasible graphs of five OPSs with the vector of the numbers of pairwise orthogonal relations (6,\,3,\,1)}\label{fig8}
		\end{center}
	\end{figure}
\begin{theorem} \label{theorem4}
	Five tripartite OPSs, any two of which are orthogonal only on one subsystem, can be perfectly distinguished by LOCC, except the following two cases:

(1) There exists a vertex ordering such that the vertices  $v_{j}$  in their orthogonality graph satisfy 
\begin{align}
&deg_{1}(v_{j})=2\,\, for\,\, j=1,\,2,\,3,\,4,\,5;\nonumber\\
&deg_{2}(v_{1})=2;\nonumber\\
&deg_{2}(v_{j})=deg_{3}(v_{j})=1\,\,for\,\, j=2,\,3,\,4,\,5.\nonumber
\end{align}

(2) Among five tripartite OPSs, there exist four states and an ordering of these four states such that the vertices $v_{j}$ of the orthogonality graph of these four OPSs satisfy  
\begin{align}
&deg_{1}(v_{j})=deg_{2}(v_{j})=deg_{3}(v_{j})=1\,\, for\,\, j=1,\,2,\,3,\,4.\nonumber
\end{align}
\end{theorem}

By analyzing all feasible orthogonality graphs of five tripartite OPSs, as shown in Figs. \ref{fig6}-\ref{fig13}, we know that graph (11-30) falls under exceptional Case 1, and graphs (9-8), (11-26), (12-15), (12-20), (13-22) and (13-27) fall under  exceptional Case 2. Therefore, except for those corresponding to these seven orthogonality graphs, any five tripartite OPSs, any two of which are orthogonal only on one subsystem, can be perfectly distinguished by LOCC. The proof of Theorem~\ref{theorem4} is given in Appendix \ref{app2}.

\begin{figure}[H]
	    \setlength{\belowcaptionskip}{0.3cm}
	    \centering
	     \includegraphics[width=0.40\textwidth]{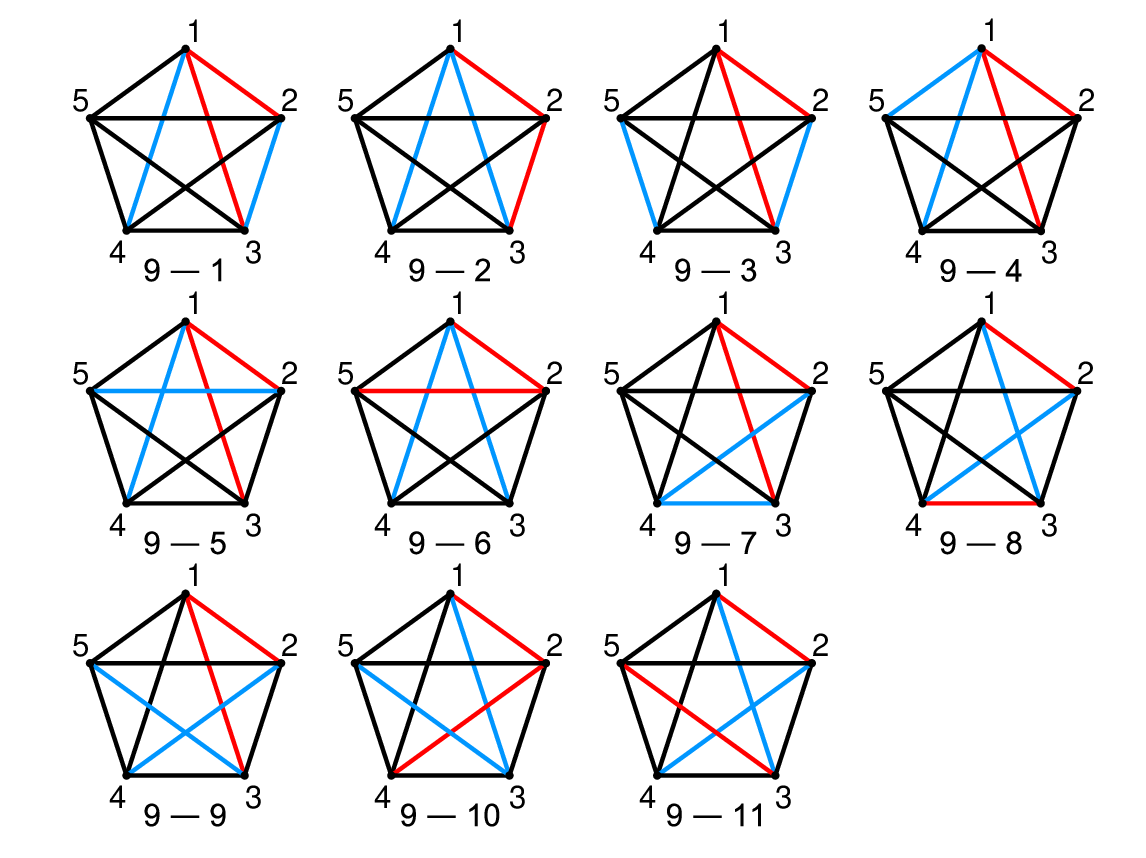}
     	\begin{center}
		\caption{The feasible graphs of five tripartite OPSs with the vector of the numbers of pairwise orthogonal relations $(6,\,2,\,2)$}\label{fig9}
		\vspace{-40pt}
	    \end{center}
    \end{figure}
    
    \begin{figure}[H]
		\setlength{\belowcaptionskip}{0.3cm}
		\centering
		\includegraphics[width=0.40\textwidth]{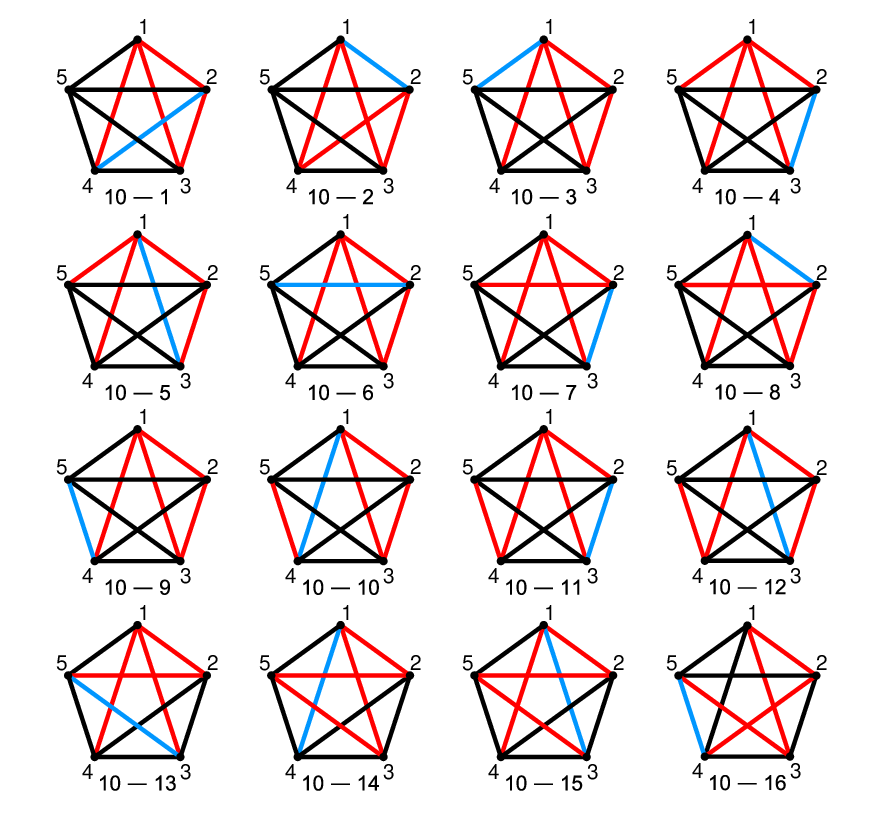}
		\begin{center}
			\caption{The feasible graphs of five OPSs with the vector of the numbers of pairwise orthogonal relations $(5,\,4,\,1)$ }\label{fig10}
			\vspace{-40pt}
		\end{center}
	\end{figure}

	\begin{figure}[H]
		\setlength{\belowcaptionskip}{0.3cm}
		\centering
		\includegraphics[width=0.425\textwidth]{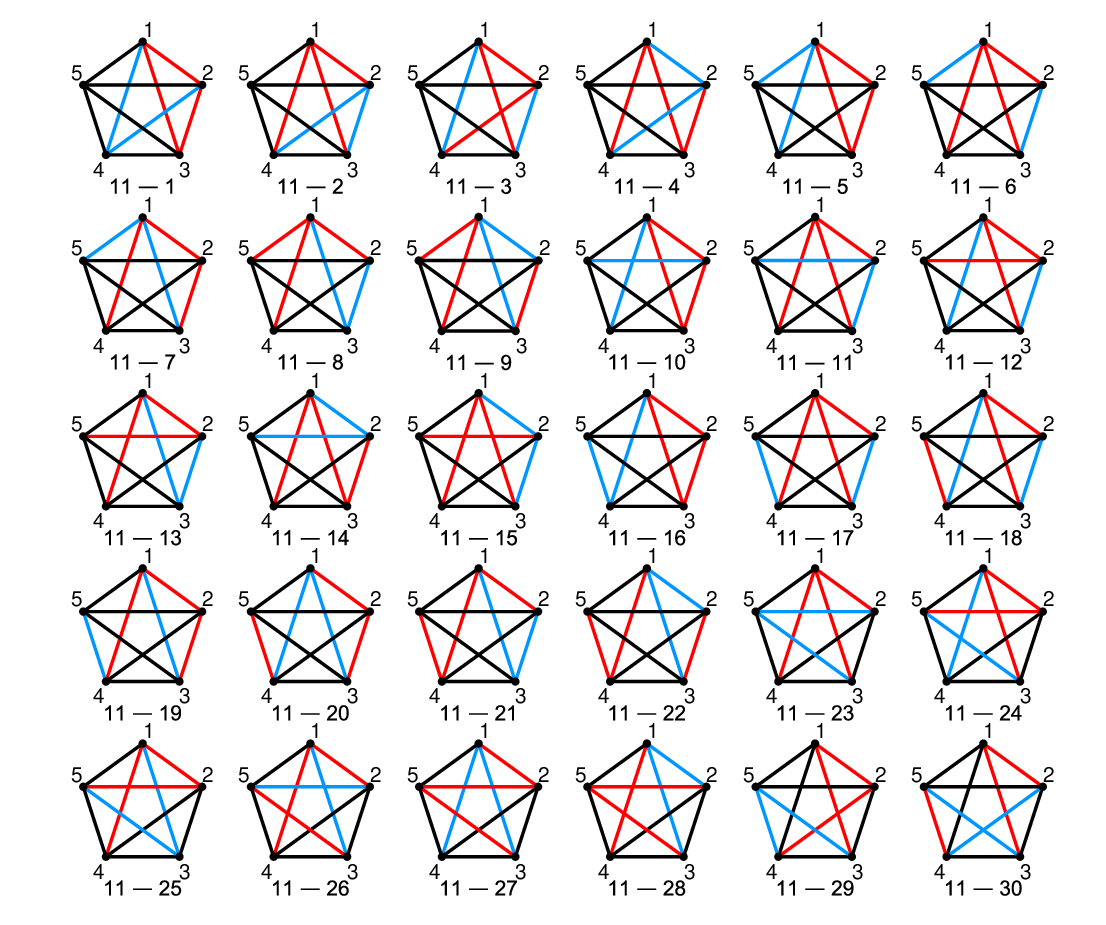}
		\begin{center}
			\caption{The feasible graphs of five tripartite OPSs with the vector of the numbers of pairwise orthogonal relations $(5,\,3,\,2)$}\label{fig11}
			\vspace{-40pt}
		\end{center}
	\end{figure}

\begin{figure}[H]
		\setlength{\belowcaptionskip}{0.3cm}
		\centering
		\includegraphics[width=0.35\textwidth]{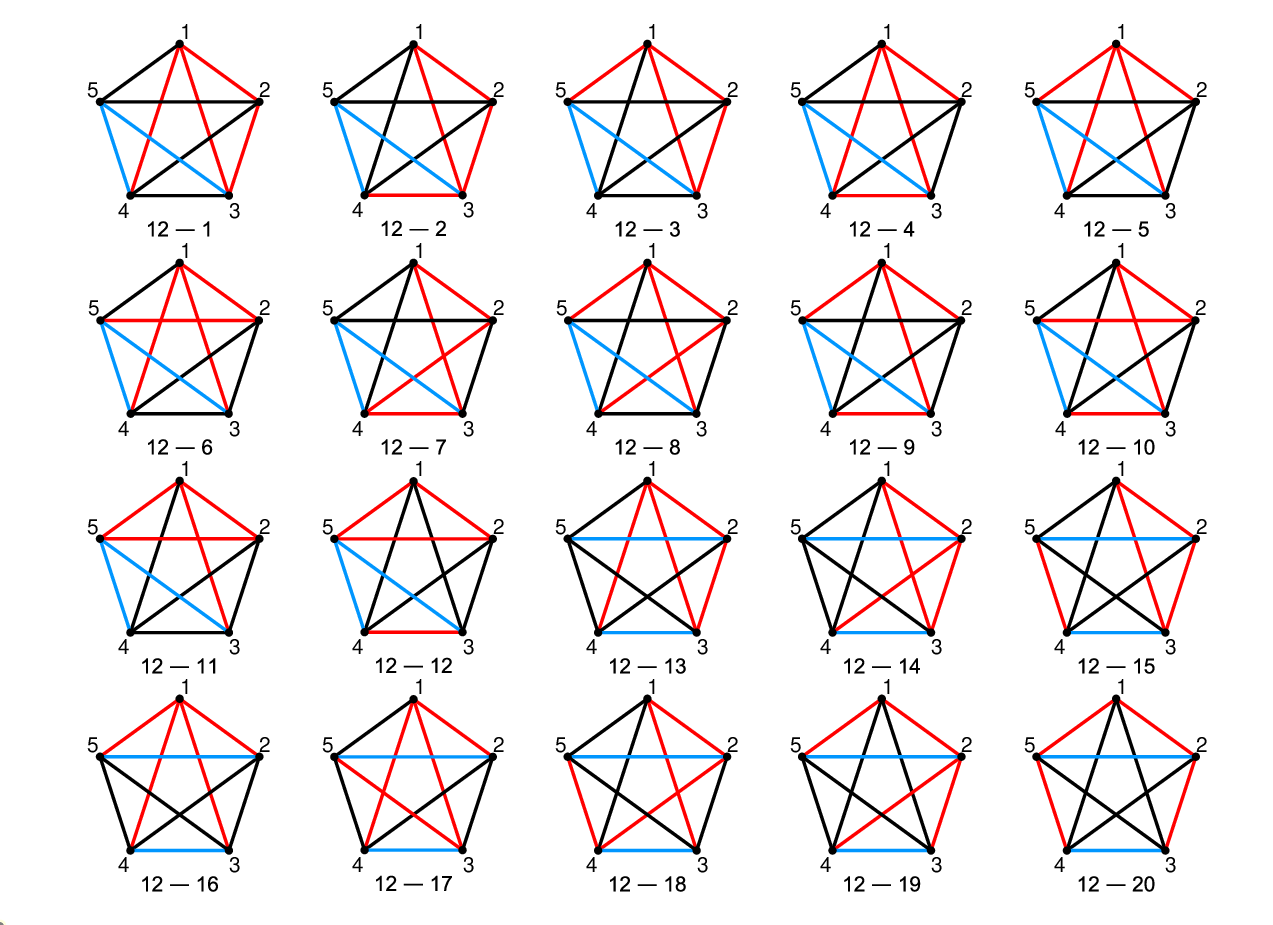}
		\begin{center}
			\caption{The feasible graphs of five OPSs with the vector of the numbers of pairwise orthogonal relations $(4,\,4,\,2)$}\label{fig12}
			\vspace{-40pt}
		\end{center}
	\end{figure}

	\begin{figure}[H]
		\setlength{\belowcaptionskip}{0.3cm}
		\centering
		\includegraphics[width=0.425\textwidth]{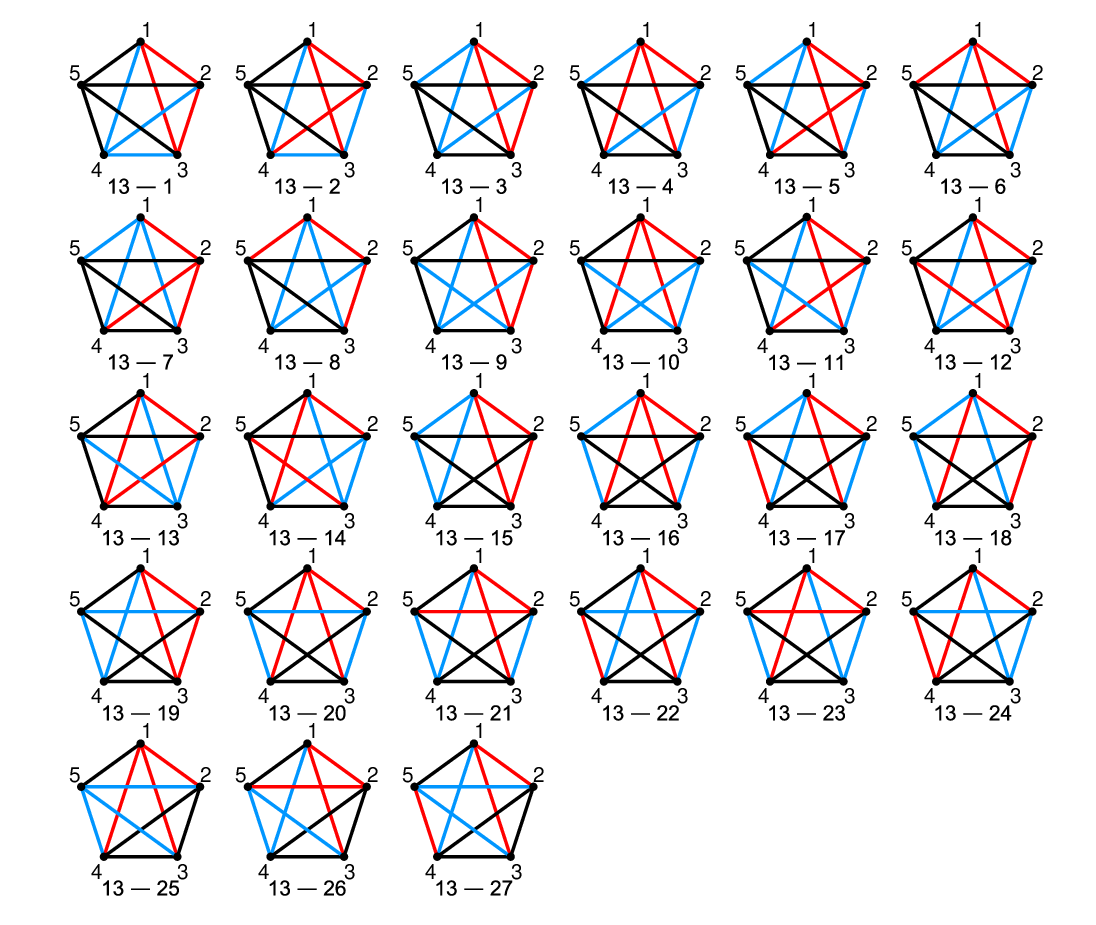}
		\begin{center}
			\caption{The feasible graphs of five tripartite OPSs with the vector of the numbers of pairwise orthogonal relations $(4,\,3,\,3)$}\label{fig13}
			\vspace{-40pt}
		\end{center}
	\end{figure} 

Now, we discuss the local distinguishability of five tripartite OPSs that correspond to each of those seven graphs. 

For excluded case 1 in Theorem~\ref{theorem4}, i.e., five tripartite OPSs that correspond to graph (11-30), as shown in Fig. 14, we construct a set of five tripartite OPSs that cannot be perfectly distinguished by LOCC and a set of five tripartite OPSs that can be perfectly distinguished by LOCC, respectively. This means that the local distinguishability of five tripartite OPSs corresponding to graph (11-30) deserves further in-depth investigation. 
\begin{figure}[H]
		\setlength{\belowcaptionskip}{0.3cm}
		\centering
		\includegraphics[width=0.2\textwidth]{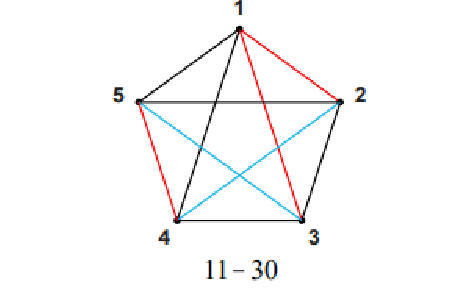}
		\begin{center}
			\caption{Excluded case 1 for Theorem 4}\label{fig14}
			\vspace{-20pt}
		\end{center}
\end{figure}

  Example 1. One set of five tripartite OPSs with orthogonality graph (11-30) is shown in Eq. (10). 
     \begin{align}
  		&|\phi_{1}\rangle=|0\rangle_{1}|0\rangle_{2}(\frac{1}{\sqrt{3}}|0\rangle+\sqrt{\frac{2}{3}}|1 \rangle)_{3},\notag\\
  		&|\phi_{2}\rangle=\frac{1}{\sqrt{6}}|0+1-2\rangle_{1}|1\rangle_{2}|0+1\rangle_{3},\notag\\
  		&|\phi_{3}\rangle=\frac{1}{\sqrt{2}}|0+2\rangle_{1}|1\rangle_{2}|0\rangle_{3},\notag\\
  		&|\phi_{4}\rangle=\frac{1}{2}|1\rangle_{1}|0+1\rangle_{2}|0-1\rangle_{3},\notag\\
  		&|\phi_{5}\rangle=\frac{1}{2}|1+2\rangle_{1}|0-1\rangle_{2}|1\rangle_{3}.
  \end{align}
  In fact, these five tripartite OPSs cannot be perfectly distinguished by LOCC. To prove the local indistinguishability of these five tripartite OPSs, we only need to show that any of the three parties can only perform a trivial POVM when it is necessary to maintain the orthogonality of post-measurement states. For these five states, it suffices to prove that neither party can obtain useful information regardless of which party performs the measurement first. The proof method can refer to that of Theorem~\ref{theorem2}.
  
  Example 2. Another set of five tripartite OPSs with orthogonality graph (11-30) is shown in Eq. (11). 
  
  \begin{align}
  		&|\phi_{1}\rangle=|0\rangle_{1}|0\rangle_{2}|0\rangle_{3},\notag\\
  		&|\phi_{2}\rangle=\frac{1}{\sqrt{6}}|0+1-2\rangle_{1}|1\rangle_{2}|0+1\rangle_{3},\notag\\
  		&|\phi_{3}\rangle=\frac{1}{2}|0+2\rangle_{1}|1\rangle_{2}|0+2\rangle_{3},\notag\\
  		&|\phi_{4}\rangle=\frac{1}{2}|1\rangle_{1}|0+1\rangle_{2}|0-1\rangle_{3},\notag\\
  		&|\phi_{5}\rangle=\frac{1}{2\sqrt{2}}|1+2\rangle_{1}|0-1\rangle_{2}|0-2\rangle_{3},
  \end{align}
  where $|0+1-2\rangle$ denotes $|0\rangle+|1\rangle-|2\rangle$, etc. These five tripartite OPSs can be perfectly distinguished by LOCC. Now we give a POVM protocol to distinguish the five states in Eq. (11). Consider the following four measurement operators
   \begin{align}
   		&\Pi_{1}=\frac{\sqrt{3}}{6}|0+1+2\rangle\langle0+1+2|,\notag\\
   		&\Pi_{2}=\frac{\sqrt{3}}{6}|0-1+2\rangle\langle0-1+2|,\notag\\
   		&\Pi_{3}=\frac{\sqrt{3}}{6}|0+1-2\rangle\langle0+1-2|,\notag\\
   		&\Pi_{4}=\frac{\sqrt{3}}{6}|0-1-2\rangle\langle0-1-2|,
   	\end{align}
   where the four operators satisfy the completeness equation, i.e.,
   $$\sum_{i=1}^{4}\Pi_{i}^{\dag}\Pi_{i}=I.$$
   Suppose that the third party, say Charlie, performs a measurement on the third subsystem of the measured state with the operators $\{\Pi_{i}:\,i=1,\,2,\,3,\,4\}$.
   
   \textcircled{1} If the measurement outcome corresponds to $\Pi_{1}$, the measured state must be $|\phi_{1}\rangle$, $|\phi_{2}\rangle$ or $|\phi_{3}\rangle$. By Lemma~\ref{lemma1}, states $|\phi_{1}\rangle$, $|\phi_{2}\rangle$ and  $|\phi_{3}\rangle$ can be perfectly identified by the first and second parties since any two are orthogonal on the first or second subsystem. 
   
   \textcircled{2} If the measurement outcome corresponds to $\Pi_{2}$, the measured state must be $|\phi_{1}\rangle$, $|\phi_{3}\rangle$ or $|\phi_{4}\rangle$. By Lemma~\ref{lemma1}, states $|\phi_{1}\rangle$, $|\phi_{3}\rangle$ and $|\phi_{4}\rangle$ can be perfectly identified by the first and second parties since any two are orthogonal on the first or second subsystem. 
   
   \textcircled{3} If the measurement outcome corresponds to $\Pi_{3}$, the measured state must be $|\phi_{1}\rangle$, $|\phi_{2}\rangle$ or $|\phi_{5}\rangle$. By Lemma~\ref{lemma1}, states $|\phi_{1}\rangle$, $|\phi_{2}\rangle$ and $|\phi_{5}\rangle$ can be perfectly identified by the first and second parties since any two are orthogonal on the first or second subsystem. 
   
   \textcircled{4} If the measurement outcome corresponds to $\Pi_{4}$, the measured state must be $|\phi_{1}\rangle$, $|\phi_{4}\rangle$ or $|\phi_{5}\rangle$. By Lemma~\ref{lemma1}, states $|\phi_{1}\rangle$, $|\phi_{4}\rangle$ and $|\phi_{5}\rangle$ can be perfectly identified by the first and second parties since any two are orthogonal on the first or second subsystem. 
   
   Therefore, states in Eq. (11) can be perfectly distinguished by LOCC.
   
   For excluded case 2 in Theorem~\ref{theorem4}, i.e., five tripartite OPSs that correspond to graph (9-8), graph (11-26), graph (12-15), graph (12-20), graph (13-22), or graph (13-27), as shown in Fig. 15, we discuss separately the local distinguishability of five tripartite OPSs corresponding to each graph.

\begin{figure}[H]
		\setlength{\belowcaptionskip}{0.3cm}
		\centering
		\includegraphics[width=0.35\textwidth]{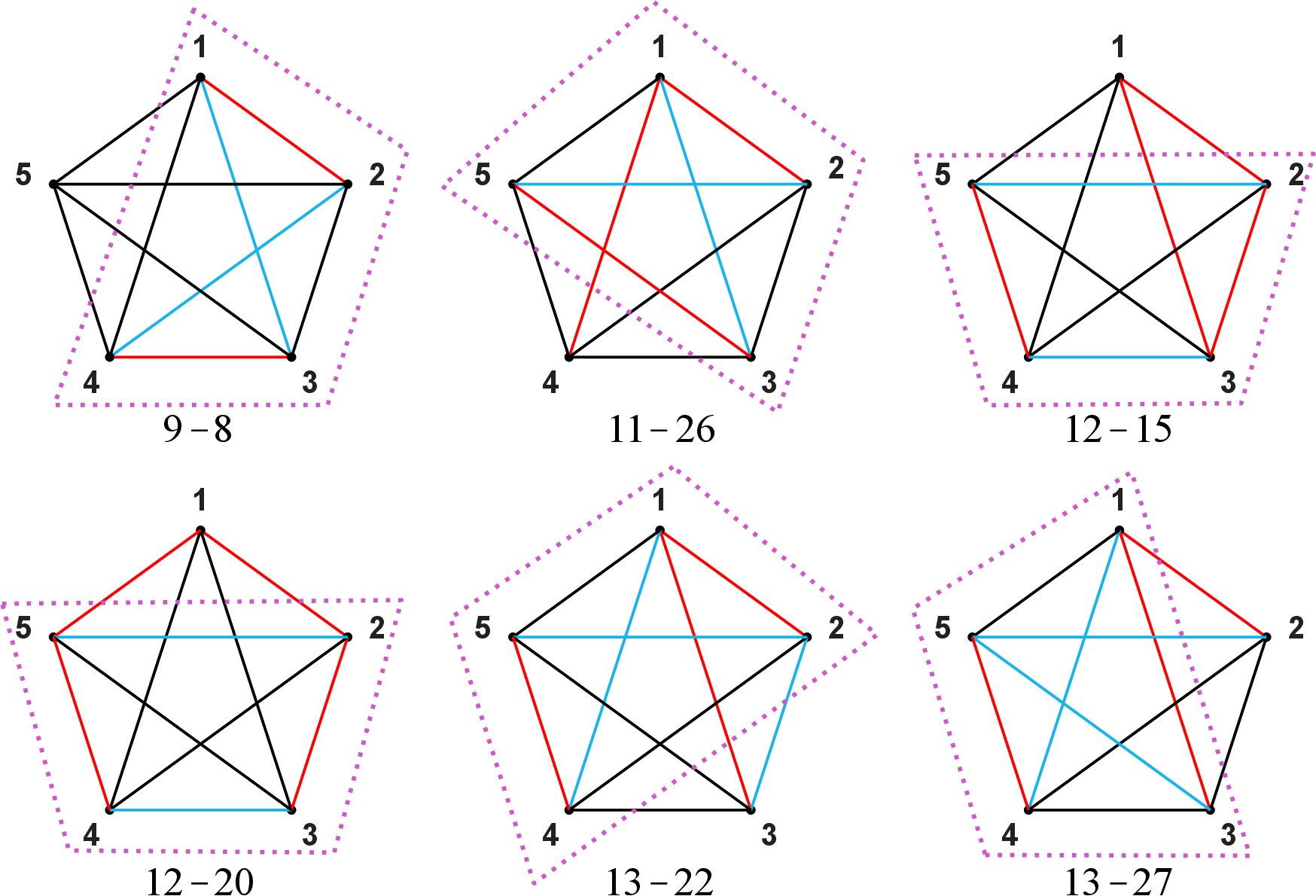}
		\begin{center}
			\caption{Excluded case 2 for Theorem 4}\label{fig15}
			\vspace{-20pt}
		\end{center}
\end{figure} 

   (1) In case (9-8), as shown in Fig. 15, state 5 is orthogonal to all the other states on the first subsystem. We assume that the first subsystem of state 5 is $|\alpha\rangle$, where $|\alpha\rangle$ is normalized. Suppose that the first party, say Alice, performs a measurement with the measurement operators $|\alpha\rangle\langle\alpha|$ and $I-|\alpha\rangle\langle\alpha|$. 

  \textcircled{1} If Alice's measurement outcome corresponds to the operator $|\alpha\rangle\langle\alpha|$, the measured state must be state 5.

  \textcircled{2} If Alice's measurement outcome corresponds to the operator $I-|\alpha\rangle\langle\alpha|$, the measured state must be one of states $\{1,\,2,\,3,\,4\}$. Note that state 1 remains orthogonal to state 4, and state 2 remains orthogonal to state 3 after Alice's measurement. The orthogonality graph corresponding to states $\{1,\,2,\,3,\,4\}$ is the part circled by the red dashed line in the first subfigure of Fig. 15, and is identical to graph (1-3) if the vertex labels are appropriately relabeled. By Lemma~\ref{lemma5}, states $\{1,\,2,\,3,\,4\}$ cannot be perfectly distinguished or can be distinguished with some certain probability by LOCC. 
  
    (2) In case (11-26), state 4 is orthogonal to any one of states $\{2,\, 3,\, 5\}$ on the first subsystem. Suppose that the first subsystem of state 4 is $|\alpha\rangle$, where $|\alpha\rangle$ is normalized. We assume that the first party, say Alice, performs a measurement with the operators $|\alpha\rangle\langle\alpha|$ and $I-|\alpha\rangle\langle\alpha|$. 
  
  \textcircled{1} If Alice's measurement corresponds to $|\alpha\rangle\langle\alpha|$, the measured state must be state 4 or 1. It can be perfectly distinguished by the second party since states 1 and 4 are orthogonal on the second subsystem. 
  
  \textcircled{2} If Alice's measurement corresponds to $I-|\alpha\rangle\langle\alpha|$, the measured state must be one of states $\{1,\,2,\, 3,\, 5\}$. Note that state 1 remains orthogonal to state 5, and state 2 remains orthogonal to state 3 after Alice's measurement. The orthogonality graph corresponding to states $\{1,\,2,\,3,\,5\}$ is the part circled by the red dashed line in the second subfigure of Fig. 15, and  is identical to graph (1-3) if the vertex labels are appropriately relabeled. By Lemma~\ref{lemma5}, states $\{1,\,2,\,3,\,5\}$ cannot be perfectly distinguished or can be distinguished with some certain probability by LOCC. 
    
      (3) In case (12-15), state 1 is orthogonal to state 5 and state 4 on the first subsystem. We assume that the first subsystem of state 1 is $|\alpha\rangle$, 
where $|\alpha\rangle$ is normalized. Suppose that the first party, say Alice, performs a measurement with the measurement operators $|\alpha\rangle \langle\alpha |$ and $I-|\alpha\rangle \langle\alpha |$. 

\textcircled{1} If Alice's measurement outcome corresponds to $|\alpha\rangle\langle\alpha |$, the measured state must be state 1, 2 or state 3. States \{1, 2, 3\} can be perfectly distinguished by LOCC since they are pairwise orthogonal on the second subsystem.

 \textcircled{2} If Alice's measurement outcome corresponds to $I-|\alpha\rangle \langle\alpha |$, the measured state must be one of states \{2, 3, 4, 5\}. Note that state 5 remains orthogonal to state 3, and state 4 remains orthogonal to state 2 after Alice's measurement. The orthogonality graph corresponding to states $\{2,\,3,\,4,\,5\}$ is the part circled by the red dashed line in the third subfigure of Fig. 15, and  is identical to graph (1-3) if the vertex labels are appropriately relabeled. By Lemma~\ref{lemma5}, states $\{2,\,3,\,4,\,5\}$ cannot be perfectly distinguished or can be distinguished with some certain probability by LOCC.

    (4) In case (12-20), state 1 is orthogonal to state 3 and state 4 on the first subsystem. We assume that the first subsystem of state 1 is $|\alpha\rangle$, 
where $|\alpha\rangle$ is normalized. Suppose that the first party, say Alice, performs a measurement with the measurement operators $|\alpha\rangle \langle\alpha |$ and $I-|\alpha\rangle \langle\alpha |$. 

\textcircled{1} If Alice's measurement outcome corresponds to $|\alpha\rangle\langle\alpha |$, the measured state must be state 1, 2 or state 5. By Lemma~\ref{lemma2}, states \{1, 2, 5\} can be perfectly distinguished by LOCC since they are pairwise orthogonal. 

\textcircled{2} If Alice's measurement outcome corresponds to $I-|\alpha\rangle \langle\alpha |$, the measured state must be one of states \{2, 3, 4, 5\}. Note that state 5 remains orthogonal to state 3, and state 4 remains orthogonal to state 2 after Alice's measurement. The orthogonality graph corresponding to states $\{2,\,3,\,4,\,5\}$ is the part circled by the red dashed line in the fourth subfigure of Fig. 15, and  is identical to graph (1-3) if the vertex labels are appropriately relabeled. By Lemma~\ref{lemma5}, states $\{2,\,3,\,4,\,5\}$ cannot be perfectly distinguished or can be distinguished with some certain probability by LOCC. 

    (5) In case (13-22), state 3 is orthogonal to states 4 and 5 on the first subsystem. We assume that the first subsystem of state 3 is $|\alpha\rangle$, where 
 $|\alpha\rangle$ is normalized. Suppose that the first party, say Alice, performs a measurement with the measurement operators $|\alpha\rangle \langle\alpha|$ and $I-|\alpha\rangle \langle\alpha|$. 
    
     \textcircled{1} If Alice's measurement outcome corresponds to the operator $|\alpha\rangle \langle\alpha|$, the measured state must be one of states \{1, 2, 3\}. States \{1, 2, 3\} can be locally distinguished by Lemma~\ref{lemma2} since they are still orthogonal after Alice's measurement. 
	
	\textcircled{2} If Alice's measurement outcome corresponds to the operator $I-|\alpha\rangle \langle\alpha|$, the measured state must be one of states 
    \{1, 2, 4, 5\}. Note that state 2 remains orthogonal to state 4, and state 1 remains orthogonal to state 5 on the second subsystem. The orthogonality graph corresponding to states $\{1,\,2,\,4,\,5\}$ is the part circled by the red dashed line in the fifth subfigure of Fig. 15, and  is identical to graph (1-3) if the vertex labels are appropriately relabeled. By Lemma~\ref{lemma5}, states $\{1,\,2,\,4,\,5\}$ cannot be perfectly distinguished or can be distinguished with some certain probability by LOCC.
     
    (6) For case (13-27), we give a method to construct a set of five tripartite OPSs that cannot be perfectly distinguished by LOCC. Consider the set of tripartite OPSs shown as follows.
        
   \begin{align}
  		&|\phi_{1}\rangle=\frac{1}{\sqrt{6}}|0+1+2\rangle_{1}|0\rangle_{2}|0+1\rangle_{3},\notag\\
  		&|\phi_{2}\rangle=|0\rangle_{1}|1\rangle_{2}|1\rangle_{3},\notag\\
  		&|\phi_{3}\rangle=|1\rangle_{1}|1\rangle_{2}|1\rangle_{3},\notag\\
  		&|\phi_{4}\rangle=\frac{1}{2}|2\rangle_{1}|0+1\rangle_{2}|0-1\rangle_{3},\notag\\
  		&|\phi_{5}\rangle=\frac{1}{\sqrt{6}}(|0\rangle+\omega|1\rangle+\omega^{2}|2\rangle)_{1}|0-1\rangle_{2}|0\rangle_{3},
  \end{align}
  where $\omega=e^{\frac{2\pi\sqrt{-1}}{3}}$.
  
    Now, we prove that the five states as shown in Eq. (13) cannot be perfectly distinguished by LOCC. Without loss of generality, we suppose that the first party, say Alice, performs a measurement with the POVM elements 
    
    $M_{j}^{\dagger}M_{j}=$ 
    $$\begin{pmatrix}
		m_{00}^{(j)} & m_{01}^{(j)} & m_{02}^{(j)} &m_{03}^{(j)} &\cdots &m_{0,n-1}^{(j)}\\
		m_{10}^{(j)} & m_{11}^{(j)} & m_{12}^{(j)} &m_{13}^{(j)} &\cdots &m_{1,n-1}^{(j)}\\
        m_{20}^{(j)} & m_{21}^{(j)} & m_{22}^{(j)} &m_{23}^{(j)} &\cdots &m_{2,n-1}^{(j)}\\
        m_{30}^{(j)} & m_{31}^{(j)} & m_{32}^{(j)} &m_{33}^{(j)} &\cdots &m_{3,n-1}^{(j)}\\
        \vdots & \vdots & \vdots & \vdots &\ddots &\vdots\\
		m_{n-1,0}^{(j)} & m_{n-1,1}^{(j)} & m_{n-1,2}^{(j)} &m_{n-1,3}^{(j)} &\cdots &m_{n-1,n-1}^{(j)} \\ 
	\end{pmatrix}$$\\
    under the basis $\{|0\rangle,\,|1\rangle,\cdots,\,|(n-1)\rangle\}$ for $j=1$, $2$, $\cdots$, $l$, where $n\geq3$, and $n$ denotes the dimension of the space in which the first subsystems of these five states reside. To ensure that the measurement can proceed, any two states that are orthogonal only on Alice's side should remain orthogonal after being measured by Alice. For states $|\phi_{2}\rangle$ and $|\phi_{3}\rangle$, we have $\langle\phi_{2}|(M_{j}^{\dagger}M_{j}\otimes I_{B}\otimes I_{C})|\phi_{3}\rangle=0$ and $\langle\phi_{3}|(M_{j}^{\dagger}M_{j}\otimes I_{B}\otimes I_{C})|\phi_{2}\rangle=0$. Thus, $m_{01}^{(j)}$=$m_{10}^{(j)}$=0. Similarly, we have $m_{02}^{(j)}$=$m_{20}^{(j)}$=0 for states $|\phi_{2}\rangle$ and $|\phi_{4}\rangle$, and $m_{12}^{(j)}$=$m_{21}^{(j)}$=0 for states $|\phi_{3}\rangle$ and $|\phi_{4}\rangle$. For states $|\phi_{1}\rangle$ and $|\phi_{5}\rangle$, we have $\langle\phi_{1}|M_{j}^{\dagger}M_{j}\otimes I_{B}\otimes I_{C}|\phi_{5}\rangle=0$ and  $\langle\phi_{5}|M_{j}^{\dagger}M_{j}\otimes I_{B}\otimes I_{C}|\phi_{1}\rangle=0$. Thus $m_{11}^{(j)}=m_{22}^{(j)}=m_{00}^{(j)}$. Consequently, each POVM element $M_{j}^{\dagger}M_{j}$ must have the form\\
    $$\begin{pmatrix}
		 m_{00}^{(j)} & 0 & 0 &m_{03}^{(j)} &\cdots &m_{0,n-1}^{(j)}\\
		 0 & m_{00}^{(j)} & 0 &m_{13}^{(j)} &\cdots &m_{1,n-1}^{(j)}\\
         0 & 0 & m_{00}^{(j)} &m_{23}^{(j)} &\cdots &m_{2,n-1}^{(j)}\\
         m_{30}^{(j)} & m_{31}^{(j)} & m_{32}^{(j)} &m_{33}^{(j)} &\cdots &m_{3,n-1}^{(j)}\\
         \vdots & \vdots & \vdots & \vdots &\ddots &\vdots\\
		 m_{n-1,0}^{(j)} & m_{n-1,1}^{(j)} & m_{n-1,2}^{(j)} &m_{n-1,3}^{(j)} &\cdots &m_{n-1,n-1}^{(j)} \\ 
	\end{pmatrix}.$$\\
    Thus, when Alice's measurement outcome is $j$, the probability that the outcome $j$ occurs is $p(j)=\langle\phi_{k}|M_{j}^{\dag}M_{j}\otimes I_{B}\otimes I_{C})|\phi_{k}\rangle=m_{00}^{(j)}$ for $k\in\{1,\,2,\,3,\,4,\,5\}$, where $j\in\{1,\,$ $2,\,$ $\cdots,\,$ $l\}$. This means that the probability of measurement outcome $j$ occurring is identical for any state $|\phi_{k}\rangle$. Therefore, Alice cannot get any useful information to distinguish these five OPSs.  In fact, the second party and the third party, will face the similar situation as Alice does. Therefore, these five states cannot be perfectly distinguished by LOCC. This means that there exists a nonlocal set of five tripartite OPSs with orthogonality graph (13-27).
   
   From the above analysis, we conclude that after removing the graphs that are identical to one of the given graphs under the exchange of parties, five tripartite OPSs give rise to a total of 124 distinct orthogonality graphs. All of them are locally distinguishable, except for the seven cases corresponding to the orthogonality graphs excluded by Theorem~\ref{theorem4}.

\section{Conclusions}\label{sec5}
Local distinguishability of orthogonal quantum states, which can reduce quantum communication and lower costs, is an important research topic in the field of quantum information. Up to now, most of the achievements made in this field are about the methods to construct nonlocal sets of OPSs \cite{HQ2023,Xu2023,Zhu2023}. 
In Ref. \cite{XUHWYJ2025}, local distinguishability of four OPSs on multipartite quantum system was given, and a natural question, ``what the local distinguishability of a set of five OPSs is on bipartite and multipartite systems" arises. In fact, compared with the local distinguishability of four OPSs, local distinguishability of five OPSs is more complex. 

In this paper, we discuss local distinguishability of five OPSs on bipartite and tripartite quantum systems, respectively. We prove that any five bipartite OPSs are locally distinguishable except a case with orthogonality graph (5-4). Meanwhile, we conclude that any five tripartite OPSs can be locally distinguished by LOCC except two cases corresponding to orthogonality graphs (9-8), (11-26), (11-30), (12-15), (12-20), (13-22), (13-27). Furthermore, we discuss the local distinguishability of five OPSs of these exceptional scenarios. Our work provides a new understanding of the local distinguishability of five OPSs on bipartite systems and tripartite systems. 

\begin{acknowledgments}
	\vspace{-10pt}
    This work is supported by the Natural Science Foundation of Shandong Province of China (Grant No. ZR2023MF080) and Beijing Natural Science Foundation (Grant No. 4252014).
\end{acknowledgments}

\appendix
\section{The proof of Theorem 1}\label{app1}
\begin{proof}
    We prove that five bipartite OPSs, any two of which are orthogonal only on one subsystem, can be perfectly distinguished by LOCC except the case with orthogonality graph (5-4).
           
    1.  Category (9,\,1)

   The vector (9,\,1) indicates that five bipartite OPSs have nine and one pairwise orthogonal relations on the first and second subsystems, respectively. As shown in orthogonality graph (2-1) of Fig.~\ref{fig2}, the states 1 and 2 are orthogonal on the second subsystem while any other two of states \{1, 2, 3, 4, 5\} are orthogonal on the first subsystem. The second party can identify states \{1, 3, 4, 5\} or states \{2, 3, 4, 5\} since states 1 and 2 are orthogonal on the second subsystem.  Note that any two elements of states \{1, 2, 3, 4\} or \{2, 3, 4, 5\} are orthogonal on the first subsystem. Therefore, both states \{1, 2, 3, 4\} and states \{2, 3, 4, 5\} can be locally distinguished by the first party. 
   
   2. Category (8,\,2) 
   
    The orthogonal graphs of five bipartite OPSs with the vector of the numbers of pairwise orthogonal relations (8,\,2) are shown in Fig.~\ref{fig2}. There exist two different cases, i.e., (2-2) and (2-3),  in Fig.~\ref{fig2}. 

    In cases (2-2) and (2-3), state 5 is orthogonal to all the other states on the first subsystem. Suppose that the first subsystem of state 5 is  $|\alpha\rangle$, where $|\alpha\rangle$ is normalized. The first party, say Alice, can perform a POVM on the first subsystem of the measured state with the operators  $|\alpha\rangle\langle\alpha|$ and  $I-|\alpha\rangle\langle\alpha|$, where $I$ is the identity operator. If Alice's measurement outcome  corresponds to $|\alpha\rangle\langle\alpha|$, the measured state must be state 5. Otherwise, if Alice's measurement outcome corresponds to $I-|\alpha\rangle\langle\alpha|$, the measured state must be one of states $\{1,\,2,\,3,\,4\}$. Obviously, states $\{1,\,2,\,3,\,4\}$ are still pairwise orthogonal after Alice's measurement. By Lemma~\ref{lemma1}, states $\{1,\,2,\,3,\,4\}$ can be perfectly distinguished by LOCC.
   
  3. Category  (7,\,3)

     The orthogonal graphs of five OPSs with the vector of the numbers of pairwise orthogonal relations (7,\,3) have four different cases, i.e., (3-1), (3-2), (3-3) and (3-4), as shown in Fig.~\ref{fig3}.

   (1) In cases (3-1), (3-2) and (3-3), state 5 is orthogonal to all the other states on the first subsystem. Suppose that the first subsystem of state 5 is  $|\alpha\rangle$, where $|\alpha\rangle$ is normalized. The first party, say Alice, can perform a POVM on the first subsystem of the measured state with the operators  $|\alpha\rangle\langle\alpha|$ and $I-|\alpha\rangle\langle\alpha|$, where $I$ is the identity operator. If Alice's measurement outcome corresponds to $|\alpha\rangle\langle\alpha|$, the measured state must be state 5. Otherwise, if Alice's measurement outcome corresponds to $I-|\alpha\rangle\langle\alpha|$, the measured state must be one of states $\{1,\,2,\,3,\,4\}$. Obviously, states $\{1,\,2,\,3,\,4\}$ are still pairwise orthogonal after Alice's measurement. By Lemma~\ref{lemma1}, the measured state can be perfectly identified by LOCC.

   (2) In case (3-4), state 4 is orthogonal to each of states $ \{1,\, 2, \,3\}$ on the first subsystem. We assume that the first subsystems of states 1, 2, 3, 4 and 5 are  $|\alpha\rangle$, $|\beta\rangle$, $|\gamma\rangle$, $|\delta\rangle$ and $|\theta\rangle$, respectively, where $|\alpha\rangle$, $|\beta\rangle$, $|\gamma\rangle$, $|\delta\rangle$ and $|\theta\rangle$ are all normalized. It is easy to see that $\langle\delta|\alpha\rangle=\langle\theta|\alpha\rangle=\langle\gamma|\beta\rangle=\langle\delta|\beta\rangle=\langle\theta|\beta\rangle=\langle\delta|\gamma\rangle=\langle\theta|\gamma\rangle=0$.

   Suppose that the first party, say Alice, performs a POVM on the first subsystem of the measured state with the operators $|\delta\rangle \langle\delta |$ and $I-|\delta\rangle \langle\delta |$.
  
   \textcircled{1} If Alice's measurement outcome corresponds to the operator $|\delta\rangle \langle\delta |$, the measured state must be state 4 or state 5, and can be exactly identified by the second party since state 4 and state 5 are orthogonal on the second subsystem. 
   
   \textcircled{2} If Alice's measurement outcome corresponds to the operator $I-|\delta\rangle \langle\delta |$, the measured state must be state 1, state 2, state 3 or state 5. Note that we need to verify that states \{1, 2, 3, 5\} are pairwise orthogonal after Alice's measurement. When Alice's measurement outcome corresponds to $I-|\delta\rangle \langle\delta |$, the first subsystems of states \{1, 2, 3, 5\} become the forms as follows, i.e., 
   $(I-|\delta\rangle \langle\delta |)|\alpha\rangle=|\alpha\rangle$,  $(I-|\delta\rangle \langle\delta |)|\beta\rangle=|\beta\rangle$, $(I-|\delta\rangle \langle\delta |)|\gamma\rangle=|\gamma\rangle$ and $(I-|\delta\rangle \langle\delta |)|\theta\rangle=|\theta\rangle-\langle\delta|\theta\rangle|\delta\rangle$.  
   Thus, we have 
   \begin{align}\notag \\
    ((I-|\delta\rangle \langle\delta |)|\theta\rangle, |\alpha\rangle)=\langle\theta|\alpha\rangle-\langle\delta|\theta\rangle^{*}\langle\delta|\alpha\rangle=0, \notag \\
    ((I-|\delta\rangle \langle\delta |)|\theta\rangle, |\beta\rangle)=\langle\theta|\beta\rangle-\langle\delta|\theta\rangle^{*}\langle\delta|\beta\rangle=0, \notag\\
    ((I-|\delta\rangle \langle\delta |)|\theta\rangle, |\gamma\rangle)=\langle\theta|\gamma\rangle-\langle\delta|\theta\rangle^{*}\langle\delta|\gamma\rangle=0, \notag 
   \end{align}
   where $\langle\delta|\theta\rangle^{*}$ denotes the complex conjugate of $\langle\delta|\theta\rangle$. 
   This means that states \{1, 2, 3, 5\} remain pairwise orthogonal when Alice's measurement outcome corresponds to the operator $I-|\delta\rangle \langle\delta |$.  
   By Lemma~\ref{lemma1}, states \{1, 2, 3, 5\} can be perfectly distinguished by LOCC.
   
       4. Category (6,\,4) 
  
    As shown in Fig.~\ref{fig4}, there exist six different cases, i.e., (4-1), (4-2), (4-3), (4-4), (4-5) and (4-6), for this category.

     (1) In cases (4-1) and (4-5), state 5 is orthogonal to all the other states on the first subsystem. So the first party can distinguish state 5 from all the others. The result is that four states are left to be distinguished, which can be locally distinguished by Lemma~\ref{lemma1}.
     
     (2) In case (4-2), state 1 is orthogonal to states 4 and 5 on the first subsystem. We assume that the first subsystems of states 1, 2, 3, 4 and 5 are  $|\alpha\rangle$, $|\beta\rangle$, $|\gamma\rangle$, $|\delta\rangle$ and $|\theta\rangle$, respectively, where $|\alpha\rangle$, $|\beta\rangle$, $|\gamma\rangle$, $|\delta\rangle$ and $|\theta\rangle$ are all normalized. It is easy to see that $\langle\alpha|\delta\rangle=\langle\alpha|\theta\rangle=\langle\beta|\delta\rangle=\langle\beta|\theta\rangle=\langle\gamma|\delta\rangle=\langle\gamma|\theta\rangle=0$.
     Suppose that the first party, say Alice, performs a POVM on the first subsystem of the measured state with the measurement operators  $|\alpha\rangle\langle\alpha |$ and $I-|\alpha\rangle\langle\alpha|$. 
     
      \textcircled{1} If the measurement outcome corresponds to the operator $|\alpha\rangle\langle\alpha |$, the measured state must be state 1, state 2 or state 3. States \{1, 2, 3\} can be exactly identified by the second party since they are orthogonal on the second subsystem. 
      
     \textcircled{2} If the measurement outcome corresponds to the operator $I-|\alpha\rangle\langle\alpha|$, the measured state must be one of states \{2, 3, 4, 5\}. Note that we need to verify that any two of states \{2, 3, 4, 5\} are still orthogonal after Alice's measurement. After Alice's measurement, the first subsystems of states \{2, 3, 4, 5\} become $(I-|\alpha\rangle\langle\alpha|)|\beta\rangle$, $(I-|\alpha\rangle\langle\alpha|)|\gamma\rangle$, $(I-|\alpha\rangle\langle\alpha|)|\delta\rangle$ and $(I-|\alpha\rangle\langle\alpha|)|\theta\rangle$, respectively. Thus, we have the inner products as follows. 
     \begin{align}\notag \\
     &((I-|\alpha\rangle\langle\alpha|)|\beta\rangle, (I-|\alpha\rangle\langle\alpha|)|\delta\rangle)\notag \\
     =&(|\beta\rangle-\langle\alpha|\beta\rangle|\alpha\rangle, |\delta\rangle)\notag \\
     =&\langle\beta|\delta\rangle-\langle\alpha|\beta\rangle^{*}\langle\alpha|\delta\rangle  \notag\\
     =&0;\notag \\
     &((I-|\alpha\rangle\langle\alpha|)|\gamma\rangle, (I-|\alpha\rangle\langle\alpha|)|\delta\rangle)\notag\\
     =&(|\gamma\rangle-\langle\alpha|\gamma\rangle|\alpha\rangle, |\delta\rangle)\notag\\
     =&\langle\gamma|\delta\rangle-\langle\alpha|\gamma\rangle^{*}\langle\alpha|\delta\rangle\notag\\
     =&0; \notag \\
            &((I-|\alpha\rangle\langle\alpha|)|\beta\rangle, (I-|\alpha\rangle\langle\alpha|)|\theta\rangle)\notag \\
      =&(|\beta\rangle-\langle\alpha|\beta\rangle|\alpha\rangle, |\theta\rangle)\notag \\
      =&\langle\beta|\theta\rangle-\langle\alpha|\beta\rangle^{*}\langle\alpha|\theta\rangle\notag \\
      =&0; \notag 
      \end{align}
     \begin{align}\notag \\   
       &((I-|\alpha\rangle\langle\alpha|)|\gamma\rangle, (I-|\alpha\rangle\langle\alpha|)|\theta\rangle)\notag \\ 
      =&(|\gamma\rangle-\langle\alpha|\gamma\rangle|\alpha\rangle, |\theta\rangle)\notag \\ 
      =&\langle\gamma|\theta\rangle-\langle\alpha|\gamma\rangle^{*}\langle\alpha|\theta\rangle \notag \\ 
      =&0. \notag 
      \end{align} 
      This means that the orthogonality relations of states 2, 3, 4 and 5 on the first subsystem remain unchanged when Alice's measurement outcome corresponds to the operator $I-|\alpha\rangle\langle\alpha|$. On the other hand, states 2 and 3 still remain orthogonal on the second subsystem. So do states 4 and 5. Therefore, states \{2, 3, 4, 5\} are pairwise orthogonal when Alice's measurement outcome corresponds to $I-|\alpha\rangle\langle\alpha|$. By Lemma~\ref{lemma1}, states \{2, 3, 4, 5\} can be perfectly distinguished by LOCC.
                                   
     (3) In case (4-3), state 1 is orthogonal to all the other states on the second subsystem. So the second party can distinguish state 1 from all the others. The result is that four states are left to be distinguished, which can be locally distinguished by Lemma~\ref{lemma1}.
      
     (4) In case (4-4), state 1 is orthogonal to state 2, state 3 and state 4 on the second subsystem. We assume that the second subsystem of state 1 is  $|\alpha\rangle$, where $|\alpha\rangle$ is normalized. Suppose that the second party, say Bob, performs a measurement with the measurement operators  $|\alpha\rangle \langle\alpha |$ and $I-|\alpha\rangle \langle\alpha |$. 
     
      \textcircled{1} If the measurement outcome corresponds to the operator $|\alpha\rangle \langle\alpha |$, the measured state must be state 1 or 5. State 1 and 5 can be exactly identified by the first party since they are orthogonal on the first subsystem. 
      
     \textcircled{2} If the measurement outcome corresponds to the operator $I-|\alpha\rangle \langle\alpha |$, the measured state must be one of states $\{2,\,3,\,4,\,5\}$. Note that we need to verify whether state 2 remains orthogonal to state 5 after Bob's measurement. We assume that the second subsystems of state 2 and state 5 are $|\beta \rangle$ and $|\gamma\rangle$, respectively, where $|\beta \rangle$ and $|\gamma\rangle$ are all normalized. By orthogonal graph (4-4), we have $\langle\beta|\gamma\rangle=0$ and $\langle\beta|\alpha\rangle=0$. 
     The postmeasurement states  $(I-|\alpha\rangle \langle\alpha |)|\beta \rangle$ and $( I-|\alpha\rangle \langle\alpha |)|\gamma \rangle$ are orthogonal on the second subsystem since $\langle\beta |(I-|\alpha\rangle \langle\alpha |)^{\dagger }(I-|\alpha\rangle \langle\alpha|)|\gamma \rangle=\langle\beta|\gamma\rangle-\langle\beta|\alpha\rangle\langle\alpha|\delta\rangle=0$. This indicates that state 2 and state 5 remain orthogonal after Bob's measurement. States \{2, 3, 4, 5\} can be locally distinguished by Lemma~\ref{lemma1} since they are mutually orthogonal after Bob's measurement.
     
     (5) In case (4-6), state 4 is orthogonal to states $\{ 1, \,3 ,\, 5\}$ and  state 5 is orthogonal to $\{1, \,2,\,4\}$ on the first subsystem. We assume that the first subsystems of states 4 and 5 are $|\alpha\rangle$ and $|\alpha^{\perp }\rangle$, respectively, where $|\alpha\rangle$ and $|\alpha^{\perp }\rangle$ are  normalized, and $\langle\alpha|\alpha^{\bot}\rangle=0$. Suppose that the first party, say Alice, performs a measurement with the operators  $|\alpha\rangle \langle\alpha |$, $|\alpha^{\perp} \rangle \langle\alpha^{\perp } |$  and $I-|\alpha\rangle \langle\alpha |- |\alpha^{\perp} \rangle \langle\alpha^{\perp } |$.
     
     \textcircled{1} If the measurement outcome corresponds to the operator $|\alpha\rangle \langle\alpha |$, the measured state must be state 2 or 4. It can be exactly identified by the second party since states 2 and 4 are orthogonal on the second subsystem. 
     
     \textcircled{2} If the measurement outcome corresponds to the operator $|\alpha^{\perp} \rangle \langle\alpha^{\perp } |$, the measured state must be state 3 or state 5. It can be exactly identified by the second party since state 3 and state 5 are orthogonal on the second subsystem. 
     
    \textcircled{3} If the measurement outcome corresponds to the operator $I-|\alpha\rangle \langle\alpha |- |\alpha^{\perp} \rangle \langle\alpha^{\perp } |$, the measured state must be state 1, 2 or 3. Note that we need to verify whether state 2 remains orthogonal to state 3 after Alice's measurement. We assume that the first subsystems of states 2 and 3 are $|\beta \rangle$ and $|\gamma  \rangle$, respectively, where $\langle\beta|\gamma\rangle=0$, $\langle\alpha|\gamma\rangle=0$ and $\langle\beta|\alpha^{\bot}\rangle=0$. The postmeasurement states  $(I-|\alpha\rangle \langle\alpha |- |\alpha^{\perp} \rangle \langle\alpha^{\perp }|)|\beta \rangle$ and $(I-|\alpha\rangle \langle\alpha |- |\alpha^{\perp} \rangle \langle\alpha^{\perp }|)|\gamma \rangle$ are orthogonal on the first subsystem since $\langle \beta |(I-|\alpha\rangle \langle\alpha |- |\alpha^{\perp} \rangle \langle\alpha^{\perp }|)^{\dagger } (I-|\alpha\rangle\langle\alpha|- |\alpha^{\perp}\rangle\langle\alpha^{\perp }|)|\gamma\rangle=0$. This indicates that state 2 and state 3 remain orthogonal after Alice's measurement. States \{1, 2, 3\} can be locally distinguished by Lemma~\ref{lemma1} since they are still mutually orthogonal after Alice's measurement.

    5. Category (5,\,5)
    
    In this category, case (5-1), (5-2) and (5-3) can be perfectly distinguished by LOCC.
    
    (1) In case (5-1), state 5 is orthogonal to all the other states on the first subsystem. So the first party can distinguish this state from all the others. The result is that four states $\{1,\,2,\,3,\,4\}$ are left to be distinguished, which can be locally distinguished by Lemma~\ref{lemma1}.
    
    (2) In case (5-2), state 4 is orthogonal to states \{ 2, \,3, \, 5\} and state 5 is orthogonal to states \{1,\,3,\, 4\}  on the first subsystem. We assume that the first subsystems of state 4 and state 5 are $|\alpha\rangle$ and $|\alpha^{\perp }\rangle$, respectively, where $|\alpha\rangle$ and $|\alpha^{\perp }\rangle$ are all normalized, and $\langle\alpha|\alpha^{\bot}\rangle=0$. Suppose that the first party, say Alice, performs a measurement with the operators  $|\alpha\rangle \langle\alpha |$, $|\alpha^{\perp} \rangle \langle\alpha^{\perp } |$ and $I-|\alpha\rangle \langle\alpha |- |\alpha^{\perp} \rangle \langle\alpha^{\perp } |$. 
    
   \textcircled{1} If the measurement outcome corresponds to the operator $|\alpha\rangle \langle\alpha |$, the measured state must be state 1 or state 4. It can be exactly identified by the second party since state 1 and state 4 are orthogonal on the second subsystem.
    
    \textcircled{2} If the measurement outcome corresponds to the operator $|\alpha^{\perp} \rangle \langle\alpha^{\perp }|$, the measured state must be state 2 or state 5. It can be exactly identified by the second party since state 2 and state 5 are orthogonal on the second subsystem.
    
   \textcircled{3} If the measurement outcome corresponds to the operator $I-|\alpha\rangle \langle\alpha |- |\alpha^{\perp} \rangle \langle\alpha^{\perp } |$, the measured state must be state 1, 2 or 3. States \{1, 2, 3\} can be locally distinguished by Lemma~\ref{lemma1} since they are mutually orthogonal on the second subsystem.
    
    (3) In case (5-3), state 5 is orthogonal to state 1, state 2 and state 3 on the first subsystem. We assume that the first subsystem of state 5 is  $|\alpha\rangle$, where $|\alpha\rangle$ is normalized. Suppose that the first party, say Alice, performs a measurement with the operators  $|\alpha\rangle \langle\alpha |$ and $I-|\alpha\rangle \langle\alpha |$. 
    
    \textcircled{1} If the measurement outcome corresponds to the operator $|\alpha\rangle \langle\alpha |$, the measured state must be state 4 or state 5. It can be exactly identified by the second party since state 4 and state 5 are orthogonal on the second subsystem.
     
    \textcircled{2} If the measurement outcome corresponds to the operator $I-|\alpha\rangle \langle\alpha |$, the measured state must be one of states \{ 1,\, 2,\, 3,\, 4\}. Note that we need to verify whether state 4 remains orthogonal to state 2 and 3 after Alice's measurement. We assume that the first subsystems of state 2, 3 and 4 are $|\beta \rangle$, $|\gamma \rangle$ and $|\delta \rangle$, respectively, where $|\beta \rangle$, $|\gamma \rangle$ and $|\delta \rangle $ are all normalized. By orthogonality graph (5-3), we have $\langle\beta|\delta\rangle=0$, $\langle\beta|\alpha\rangle=0$, $\langle\gamma|\delta\rangle=0$ and $\langle\gamma|\alpha\rangle=0.$ The postmeasurement states  $(I-|\alpha\rangle \langle\alpha |)|\beta \rangle$ and $(I-|\alpha\rangle \langle\alpha|)|\delta \rangle$ are orthogonal on the first subsystem, since $\langle \beta |(I-|\alpha\rangle \langle\alpha |)^{\dagger } (I-|\alpha\rangle \langle\alpha |)|\delta  \rangle=0$. This indicates that state 2 and state 4 remain orthogonal after Alice's measurement. Similarly, state 3 and state 4 remain orthogonal after Alice's measurement. States \{1, 2, 3, 4\} can be locally distinguished by Lemma~\ref{lemma1} since they are mutually orthogonal after Alice's measurement.  
    
   Therefore, five bipartite OPSs, any two of which are orthogonal only on one subsystem, can be perfectly distinguished by LOCC except the case with orthogonality graph (5-4). This completes the proof. 
    \end{proof}  
         
    \section{The proof of Theorem 4}\label{app2}
    \begin{proof}  
    Now, we prove that five tripartite OPSs, any two of which are orthogonal only on one subsystem, can be perfectly distinguished by LOCC, except five tripartite OPSs with graph (11-30), (9-8), (11-26), (12-15), (12-20), (13-22) or (13-27).
      
    1. Category (8,\,1,\,1)
    
	The vector (8,\,1,\,1) indicates that five OPSs have eight, one and one pairwise orthogonal relations on the first, second and third subsystems, respectively. For this category, there exist two different orthogonality graphs, i.e., (6-1) and (6-2), as shown in Fig.~\ref{fig6}.
	In cases (6-1) and (6-2), state 5 is orthogonal to all the other states on the first subsystem. So the first party can distinguish this state from all the 
    others. The result is that four states $\{1,\,2,\,3,\,4\}$ are left to be distinguished. States $\{1,\,2,\,3,\,4\}$ can be locally distinguished by Lemma~\ref{lemma3} since their vector of the numbers of pairwise orthogonal relations is $(4,\,1,\,1)$. Therefore, five tripartite OPSs with the vector $(8,\,1,\,1)$ can be locally distinguishable by LOCC.
    
    2. Category (7,\,2,\,1)
    
   As shown in Fig.~\ref{fig7}, there exist six different orthogonality graphs, i.e., (7-1), (7-2), (7-3), (7-4), (7-5) and (7-6), for five tripartite OPSs with the vector of 
the numbers of pairwise orthogonal relations (7,\,2,\,1). 
	
	(1) In cases (7-1), (7-2), (7-3) and (7-4), state 5 is orthogonal to all the other states on the first subsystem. So the first party can identify this state from
    all the others. The result is that four states $\{1,\,2,\,3,\,4\}$ are left to be distinguished. States $\{1,\,2,\,3,\,4\}$ can be locally distinguished by Lemma~\ref{lemma3} since their vector of the numbers of pairwise orthogonal relations is (3,\,2,\,1).
	
	(2) In case (7-5), state 4 is orthogonal to state 5 on the third subsystem, which can be used to identify states  \{1, 2, 3, 4\} and states  \{1, 2, 3, 5\} by 
    the third party. It is obvious that the vector of the numbers of pairwise orthogonal relations of states \{1, 2, 3, 4\} is (4,\,2,\,0). This is means that there is no orthogonal relationship among states \{1, 2, 3, 4\} on their third subsystem. To distinguish  states \{1, 2, 3, 4\}, it is only necessary to consider their first and second subsystems. Thus, states \{1, 2, 3, 4\} can be seen as a set of bipartite OPSs that are pairwise orthogonal when we omit its third subsystem. By Lemma~\ref{lemma1}, states \{1, 2, 3, 4\} can be perfectly distinguished by LOCC. On the other hand, the same situation applies to states \{1, 2, 3, 5\}, since the vector of the number of pairwise orthogonal relations of these states is (4, 2, 0) as well. Therefore, five tripartite OPSs can be perfectly distinguished by LOCC for case (7-5).

   (3) In case (7-6), state 1 is orthogonal to state 3 on the third subsystem, which can be used to identify states \{2, 3, 4, 5\} and states \{1, 2, 4, 5\} by the third party. For states \{1, 2, 4, 5\}, the vector of the numbers of pairwise orthogonal relations is (4, 2, 0). For states \{2, 3, 4, 5\}, the vector of the numbers of pairwise orthogonal relations is (5, 1, 0). This means that both the set of states \{2, 3, 4, 5\} and the set of states \{1, 2, 4, 5\} can be seen a set of bipartite OPSs that are pairwise orthogonal when omitting the third subsystem.  By Lemma~\ref{lemma1}, both states \{2, 3, 4, 5\} and states \{1, 2, 4, 5\} can be perfectly distinguished by LOCC. Therefore, five tripartite OPSs in case (7-6) can be perfectly distinguished by LOCC.
   
   3. Category (6,\,3,\,1)
   
	As shown in Fig.~\ref{fig8}, there exist 12 different orthogonality graphs, i.e., (8-1), (8-2),\, \dots,\, (8-12), for five tripartite OPSs with the vector of the numbers of pairwise orthogonal relations (6,\,3,\,1). 
	
 (1) In cases (8-1), (8-2), (8-3) and (8-10), state 5 is orthogonal to all the other states on the first subsystem. So the first party can identify this state from all the others. The result is that four states \{1, 2, 3, 4\} are left to be distinguished. For each of cases (8-1), (8-2), (8-3) and (8-10), the vector of the numbers of pairwise orthogonal relations of states \{1, 2, 3, 4\} is (2, 3, 1). As is known, the local distinguishability of four OPSs with the vector (2,\,3,\,1) is identical to that of four OPSs with the vector of the numbers of pairwise orthogonal relations (3, 2, 1). By Lemma~\ref{lemma3}, states \{1, 2, 3, 4\} can be perfectly distinguished. Therefore, five states in case (8-1), (8-2), (8-3) or (8-10) can be perfectly distinguished by LOCC. 
 
 (2) In case (8-4), state 4 is orthogonal to state 5 on the third subsystem, which can be used to distinguish states  \{1, 2, 3, 4\} and states \{1, 2, 3, 5\} by the third party. For states \{1, 2, 3, 4\}, the vector of the numbers of pairwise orthogonal relations is (3,\,3,\,0). The set of states \{1, 2, 3, 4\} can be seen as a set of bipartite OPSs when the third subsystem is omitted. By Lemma~\ref{lemma1}, states \{1, 2, 3, 4\} can be perfectly distinguished by LOCC. So can states \{1, 2, 3, 5\} since their vector of the numbers of pairwise orthogonal relations is (3,\,3,\,0). Therefore, five tripartite OPSs with case (8-4) can be perfectly distinguished by LOCC. 
 
 (3) In case (8-5), state 2 is orthogonal to state 3 on the third subsystem, which can be used to distinguish states \{1, 3, 4, 5\} and states \{1, 2, 4, 5\} by the third party. For states \{1, 3, 4, 5\} and states  \{1, 2, 4, 5\}, the vectors of the numbers of pairwise orthogonal relations are all (4,\,2,\,0). Thus, both the set of states \{1, 3, 4, 5\} and the set of states  \{1, 2, 4, 5\} can be seen as a set of bipartite OPSs when the third subsystem is omitted. By Lemma~\ref{lemma1}, both states \{1, 3, 4, 5\} and states \{1, 2, 4, 5\} can be perfectly distinguished by LOCC.    
 
 (4) In case (8-6), state 1 is orthogonal to state 5 on the third subsystem, which can be used to distinguish states \{2, 3, 4, 5\} and states \{1, 2, 3, 4\} by the third party. For states \{2, 3, 4, 5\}, any two of them are orthogonal on the first subsystem. Thus, states \{2, 3, 4, 5\} can be perfectly distinguished by LOCC. On the other hand, for states \{1, 2, 3, 4\}, the vector of the numbers of pairwise orthogonal relations is (3,\,3,\,0). Thus, the set of states \{1, 2, 3, 4\} can be seen as a set of bipartite OPSs when the third subsystem is omitted. By Lemma~\ref{lemma1}, states \{1, 2, 3, 4\} can be perfectly distinguished by LOCC.

 (5) In case (8-7), state 2 is orthogonal to state 5 on the third subsystem, which can be used to identify states  \{1, 3, 4, 5\} and states \{1, 2, 3, 4\} by the third party.

 For states \{1, 3, 4, 5\}, the vector of the numbers of pairwise orthogonal relations is (4,\,2,\,0). This means that these four states can be seen as a set of bipartite OPSs when the third subsystem is omitted. By Lemma~\ref{lemma1}, states \{1, 3, 4, 5\} can be perfectly distinguished by LOCC. For states \{1, 2, 3, 4\}, the vector of the numbers of pairwise orthogonal relations is (3, 3, 0). This means that these four states can be seen as a set of bipartite OPSs when the third subsystem is omitted. By Lemma~\ref{lemma1}, states \{1, 2, 3, 4\} can be perfectly distinguished by LOCC. 
 
 (6) In case (8-8), state 1 is orthogonal to state 4 on the third subsystem, which can be used to identify states \{2, 3, 4, 5\} and states \{1, 2, 3, 5\} by the third party. For the set of states \{2, 3, 4, 5\} and the set of states \{1, 2, 3, 5\}, the vectors of the numbers of pairwise orthogonal relations are (5, 1, 0) and (3, 3, 0), respectively. Both the set of states \{2, 3, 4, 5\} and the set of states \{1, 2, 3, 5\} can be seen as a set of bipartite OPSs when the third subsystem is omitted. By Lemma~\ref{lemma1}, both states \{2, 3, 4, 5\} and states \{1, 2, 3, 5\} can be perfectly distinguished by LOCC. 
  
 (7) In case (8-9), state 1 is orthogonal to state 2 on the third subsystem, which can be used to identify states \{2, 3, 4, 5\} and states \{1, 3, 4, 5\} by the third party. The vectors of the numbers of pairwise orthogonal relations of states \{2, 3, 4, 5\} and states \{1, 3, 4, 5\} are (5, 1, 0) and (4, 2, 0), respectively. Both states \{2, 3, 4, 5\} and states \{1, 3, 4, 5\} can be seen as a set of bipartite OPSs when the third subsystem is omitted. By Lemma~\ref{lemma1}, both states \{2, 3, 4, 5\} and states \{1, 3, 4, 5\} can be locally distinguished.
 
 (8) In case (8-11), state 3 is orthogonal to state 5 on the third subsystem, which can be used to identify states \{1, 2, 4, 5\} and states \{1, 2, 3, 4\} by the third party. The vectors of the numbers of pairwise orthogonal relations of states \{1, 2, 4, 5\} and states \{1, 2, 3, 4\} are (4, 2, 0) and (3, 3, 0), respectively. Both states \{1, 2, 4, 5\} and states \{1, 2, 3, 4\} can be seen as a set of bipartite OPSs when the third subsystem is omitted. By Lemma~\ref{lemma1}, both states \{1, 2, 4, 5\} and states \{1, 2, 3, 4\} can be perfectly distinguished by LOCC.   
  
 (9) In case (8-12), state 1 is orthogonal to state 3 on the third subsystem, which can be used to identify states \{2, 3, 4, 5\} and states \{1, 2, 4, 5\} by the third party. The vectors of the numbers of pairwise orthogonal relations of states \{2, 3, 4, 5\} and states \{1, 2, 4, 5\} are all (4, 2, 0). Both states \{2, 3, 4, 5\} and states \{1, 2, 4, 5\} can be seen as a set of bipartite OPSs when the third subsystem is omitted. By Lemma~\ref{lemma1}, both states \{2, 3, 4, 5\} and states \{1, 2, 4, 5\} can be locally distinguished.
	
 In summary, five tripartite OPSs with the vector $(6, 3, 1)$ are locally distinguishable by LOCC.
 
 4. Category (6,\,2,\,2)

     As shown in Fig.~\ref{fig9}, there exist 11 different orthogonality graphs, i.e., (9-1), (9-2), \dots, (9-11), for five tripartite OPSs with the vector (6, 2, 2).  
	
    (1) In cases (9-1) and (9-2), state 5 is orthogonal to all the other states on the first subsystem. So the first party can identify this state from all the others. The result is that four states \{1, 2, 3, 4\} are left to be distinguished.
     
    For case (9-1), the vector of the numbers of pairwise orthogonal relations of states \{1, 2, 3, 4\} is $(2,\, 2,\, 2)$. The orthogonality graph of states \{1, 2, 3, 4\} is identical to that of graph (1-1) after swapping the colors of the edges corresponding to blue and black in the orthogonality graph of states \{1, 2, 3, 4\} and ignoring the labels of these four states. This means states \{1, 2, 3, 4\} in case (9-1) have the same local distinguishability. By Lemma~\ref{lemma4}, states \{1, 2, 3, 4\} can be perfectly distinguished by LOCC. 
    
    For case (9-2), the vector of the numbers of pairwise orthogonal relations of states \{1, 2, 3, 4\} is $(2,\, 2,\, 2)$.  The orthogonality graph of states \{1, 2, 3, 4\} is identical to that of graph (1-2) when ignoring the labels of these four states. By Lemma~\ref{lemma4}, states \{1, 2, 3, 4\} in case (9-2) can be perfectly distinguished by LOCC.
    
    (2) In cases (9-3), (9-4), (9-5), (9-7) and (9-9), state 1 is orthogonal to state 2 and state 3 on the second subsystem, which can be used to distinguish states \{1, 4, 5\} and \{2, 3, 4, 5\} by the second party. 
    
    By Lemma~\ref{lemma2}, states \{1, 4, 5\} can be perfectly distinguished by LOCC since these three states are pairwise orthogonal. 
    
    For States \{2, 3, 4, 5\}, we will discuss case by case, respectively. \textcircled{1} In case (9-3), the vector of the numbers of pairwise orthogonal relations of states \{2, 3, 4, 5\} is (4,\,0,\,2). This means that states \{2, 3, 4, 5\} can be seen a set of bipartite OPSs when we do not consider the second subsystems of these OPSs. By Lemma~\ref{lemma1}, states \{2, 3, 4, 5\} can be perfectly distinguished by LOCC. \textcircled{2} In case (9-4), the vector of the numbers of pairwise orthogonal relations of states \{2, 3, 4, 5\} is (6,\,0,\,0). Thus states \{2, 3, 4, 5\} can be perfectly distinguished since they are pairwise orthogonal on the first subsystem. \textcircled{3} In case (9-5), the vector of the numbers of pairwise orthogonal relations of states \{2, 3, 4, 5\} is (5,\,1,\,0). This means that states \{2, 3, 4, 5\} can be seen a set of bipartite OPSs when we do not consider the third subsystems of these OPSs. By Lemma~\ref{lemma1}, states \{2, 3, 4, 5\} can be perfectly distinguished by LOCC. \textcircled{4} In cases (9-7) and (9-9), the vector of the numbers of pairwise orthogonal relations of states \{2, 3, 4, 5\} is (4,\,2,\,0). This means that states \{2, 3, 4, 5\} can be seen a set of bipartite OPSs when we do not consider the third subsystems of these OPSs. By Lemma~\ref{lemma1}, states \{2, 3, 4, 5\} can be perfectly distinguished by LOCC.

    (3) In case (9-6), state 2 is orthogonal to state 1 and state 5 on the second subsystem, which can be used to identify states \{1, 3, 4, 5\} and states \{2, 3, 4\} by the second party. For states \{1, 3, 4, 5\}, the vector of the numbers of pairwise orthogonal relations is (4,\,0,\,2). This means that states \{1, 3, 4, 5\} can be seen a set of bipartite OPSs when we do not consider the second subsystems of these OPSs. Thus, states \{1, 3, 4, 5\} can be perfectly distinguished by LOCC according to Lemma~\ref{lemma1}. On the other hand, states \{2, 3, 4\} can be perfectly distinguished by the first party sine these three OPSs are pairwise orthogonal on the first subsystem. 

    (4) In case (9-10), state 2 is orthogonal to state 1 and state 4 on the second subsystem, which can be used to identify states \{1, 3, 4, 5\} and states \{2, 3, 5\} by the second party. 
    
    For states \{1, 3, 4, 5\}, the vector of the numbers of pairwise orthogonal relations is (4,\,0,\,2). This means that states \{1, 3, 4, 5\} can be seen a set of bipartite OPSs when we do not consider the second subsystems of these OPSs. By Lemma~\ref{lemma1}, states \{1, 3, 4, 5\} can be perfectly distinguished by LOCC. On the other hand, states \{2, 3, 5\} can be locally distinguished by Lemma~\ref{lemma2}.

    (5) In case (9-11), state 4 is orthogonal to states \{ 1,\, 3,\, 5\} and state 5 is orthogonal to states \{1, 2, 4\} on the first subsystem. We assume that the first subsystems of state 4 and state 5 are $|\alpha\rangle$ and $|\alpha^{\perp }\rangle$, respectively, where $|\alpha\rangle$ and $|\alpha^{\perp }\rangle$ are all normalized. Suppose that the first party, say Alice, performs a measurement with the measurement operators  $|\alpha\rangle \langle\alpha |$, $|\alpha^{\perp} \rangle \langle\alpha^{\perp } |$ and $I-|\alpha\rangle \langle\alpha |- |\alpha^{\perp} \rangle \langle\alpha^{\perp }|$. 

    \textcircled{1} If Alice's measurement outcome corresponds to the operator $|\alpha\rangle \langle\alpha |$, the measured state must be state 2 or state 4. It can be exactly identified by the third party since state 2 and state 4 are orthogonal on the third subsystem. 

    \textcircled{2} If Alice's measurement outcome corresponds to the operator $|\alpha^{\perp} \rangle \langle\alpha^{\perp } |$, the measured state must be state 3 or state 5. It can be exactly identified by the second party since state 3 and state 5 are orthogonal on the second subsystem. 

    \textcircled{3} If Alice's measurement outcome corresponds to the operator $I-|\alpha\rangle \langle\alpha |- |\alpha^{\perp} \rangle \langle\alpha^{\perp }|$, the measured state must be state 1, state 2 or state 3. Note that we need to verify whether state 2 remains orthogonal to state 3 after Alice's measurement. We assume that the first subsystems of state 2 and state 3 are $|\beta \rangle$ and $|\gamma \rangle$, respectively, where  $|\beta \rangle$ and $|\gamma \rangle$ are all normalized. By graph (9-11), we have $\langle\alpha| \alpha^{\perp }\rangle=0$, $\langle\beta| \gamma\rangle=0$,  $\langle\beta|  \alpha^{\perp }\rangle=0$ and  $\langle\alpha|  \gamma\rangle=0$. The post-measurement states of the first subsystems of state 2 and state 3 are 
    $(I-|\alpha\rangle \langle\alpha |- |\alpha^{\perp} \rangle \langle\alpha^{\perp }|)|\beta\rangle$=$|\beta\rangle-\langle\alpha|\beta\rangle|\alpha\rangle$ and $(I-|\alpha\rangle \langle\alpha |- |\alpha^{\perp} \rangle \langle\alpha^{\perp }|)|\gamma\rangle$=$|\gamma\rangle-\langle\alpha^{\bot}|\gamma\rangle|\alpha^{\bot}\rangle$, respectively. 
    The inner product ($|\beta\rangle-\langle\alpha|\beta\rangle|\alpha\rangle$, $|\gamma\rangle-\langle\alpha^{\bot}|\gamma\rangle|\alpha^{\bot}\rangle$)=0. 
   This means that state 2 and state 3 remain orthogonal on the first subsystem after Alice's measurement. States \{1, 2, 3\} can be locally distinguished by Lemma~\ref{lemma2} since they are mutually orthogonal.
   
    5. Category (5,\,4,\,1)

	As shown in Fig.~\ref{fig10}, there exist 16 different orthogonality graphs, i.e., (10-1), (10-2), \dots, (10-16), for five tripartite OPSs with the vector of the numbers of pairwise orthogonal relations (5,\,4,\,1).  

  (1) In cases (10-1) and (10-2), state 5 is orthogonal to all the other states on the first subsystem. So the first party can distinguish this state from all the others. The result is that four states \{1, 2, 3, 4\} are left to be distinguished. The vector of the numbers of pairwise orthogonal relations of states \{1, 2, 3, 4\} is (1, 4, 1). As is known, the set of four tripartite OPSs characterized by the vector (1, 4, 1) shares the same local distinguishability as the set characterized by the vector (4, 1, 1). By Lemma~\ref{lemma3}, states \{1, 2, 3, 4\} can be perfectly distinguished by LOCC.
 
  (2) In case (10-3), state 1 is orthogonal to state 5 on the third subsystem, which can be used to identify states \{1, 2, 3, 4\} and states \{2, 3, 4, 5\} by the third party. For states \{1, 2, 3, 4\}, the vector of the numbers of pairwise orthogonal relations is (2, 4, 0). For states \{2, 3, 4, 5\}, the vector of the numbers of pairwise orthogonal relations of states is (5, 1, 0). Thus, both states  \{1, 2, 3, 4\} and states \{2, 3, 4, 5 \} can be seen as a set of bipartite OPSs when the third subsystem is not considered. By Lemma~\ref{lemma1}, both the set of states \{1, 2, 3, 4\} and the set of states \{2, 3, 4, 5\} can be perfectly distinguished by LOCC.     
  
  (3) In cases (10-4), (10-7) and (10-11), state 2 is orthogonal to state 3 on the third subsystem, which can be used to identify states  \{1, 2, 4, 5\} and states  \{1, 3, 4, 5\} by the third party. 
 
  \textcircled{1} For both states \{1, 2, 4, 5\} and states \{1, 3, 4, 5\} in case (10-4), the vectors of the numbers of pairwise orthogonal relations are all (3, 3, 0). This means that both states \{1, 2, 4, 5\} and states \{1, 3, 4, 5\} can be seen as a set of bipartite OPSs when the third subsystem is not considered. By Lemma~\ref{lemma1}, both states \{1, 2, 4, 5\} and states \{1, 3, 4, 5\} can be perfectly distinguished by LOCC.
 
  \textcircled{2} In case (10-7), the vector of the numbers of pairwise orthogonal relations of states \{1, 2, 4, 5\} is (3, 3, 0) while that of states \{1, 3, 4, 5\} is (4, 2, 0). This means that both states \{1, 2, 4, 5\} and states \{1, 3, 4, 5\} can be seen as a set of bipartite OPSs when the third subsystem is not considered. By Lemma~\ref{lemma1}, both the set of states \{1, 2, 4, 5\} and the set of states \{1, 3, 4, 5\} can be perfectly distinguished by LOCC.
  
  \textcircled{3} In case (10-11), both the vector of the numbers of pairwise orthogonal relations of states \{1, 2, 4, 5\} and that of states \{1, 3, 4, 5\} are all (3, 3, 0). similarly, both states \{1, 2, 4, 5\} and states \{1, 3, 4, 5\} can be perfectly distinguished by LOCC by Lemma~\ref{lemma1}.
 
  (4) In cases (10-5), (10-12) and (10-15), state 1 is orthogonal to state 3 on the third subsystem, which can be used to identify states  \{1, 2, 4, 5\} and states  \{2, 3, 4, 5\} by the third party.  In case (10-5), the vector of the numbers of pairwise orthogonal relations of states \{1, 2, 4, 5\} is (3, 3, 0) and that of states \{2, 3, 4, 5\} is (5, 1, 0). In cases (10-12) and (10-15), the vector of the numbers of pairwise orthogonal relations of states \{1, 2, 4, 5\} is (3, 3, 0) and that of states \{2, 3, 4, 5\} is (4, 2, 0). For the set of states \{1, 2, 4, 5\} , regardless of which case it belongs to, it can be regarded as a set of bipartite OPSs that are pairwise orthogonal when the third subsystem is not considered. By Lemma~\ref{lemma1}, states  \{1, 2, 4, 5\} can be perfectly distinguished by LOCC. So can states \{2, 3, 4, 5\}.
 
  (5) In case (10-6), state 2 is orthogonal to state 5 on the third subsystem, which can be used to identify states  \{1, 2, 3, 4\} and states \{1, 3, 4, 5\} by the third party. The vector of the numbers of pairwise orthogonal relations of states \{1, 2, 3, 4\} is (2, 4, 0) and that of states \{1, 3, 4, 5\} is (4, 2, 0). Thus both the set of states \{1, 2, 3, 4\}  and the set of states \{1, 3, 4, 5\} can be seen as a set of bipartite OPSs that are pairwise orthogonal when the third subsystem is not considered.  By Lemma~\ref{lemma1}, both states \{1, 2, 3, 4\} and states \{1, 3, 4, 5\} can be perfectly distinguished by LOCC. 
 
  (6) In case (10-8), state 1 is orthogonal to state 2 on the third subsystem, which can be used to identify states  \{1, 3, 4, 5\} and states  \{2, 3, 4, 5\} by the third party. Both the vector of the numbers of pairwise orthogonal relations of states \{1, 3, 4, 5\} and that of states \{2, 3, 4, 5\} are all (4, 2, 0). Thus both the set of states \{1, 3, 4, 5\}  and the set of states \{2, 3, 4, 5\} can be seen as a set of bipartite OPSs that are pairwise orthogonal when the third subsystem is not considered. By Lemma~\ref{lemma1}, both states \{1, 3, 4, 5\} and states \{2, 3, 4, 5\} can be perfectly distinguished by LOCC. 
   
  (7) In cases (10-9) and (10-16), state 4 is orthogonal to state 5 on the third subsystem, which can be used to identify states  \{1, 2, 3, 4\} and states  \{1, 2, 3, 5\} by the third party. For case (10-9), the vector of the numbers of pairwise orthogonal relations of states \{1, 2, 3, 4\} is (2, 4, 0) and that of states is \{1, 2, 3, 5\} is (3, 3, 0). For case (10-16), both the vector of the numbers of pairwise orthogonal relations of states \{1, 2, 3, 4\} and that of states \{1, 2, 3, 5\} are all (3, 3, 0). Thus both the set of \{1, 2, 3, 4\} and the set of states \{1, 2, 3, 5\}, regardless of which cases the two sets belong to, can be seen as a set of bipartite OPSs that are pairwise orthogonal when the third subsystem is not considered. By Lemma~\ref{lemma1}, both states \{1, 2, 3, 4\} and states \{1, 2, 3, 5\} can be perfectly distinguished by LOCC. 
  
  (8) In cases (10-10) and (10-14), state 1 is orthogonal to state 4 on the third subsystem, which can be used to identify states \{1, 2, 3, 5\} and states \{2, 3, 4, 5\} by the third party. For case (10-10), the vector of the numbers of pairwise orthogonal relations of states \{1, 2, 3, 5\} is (3, 3, 0) and that of states \{2, 3, 4, 5\} is (4, 2, 0). For case (10-14), the vector of the numbers of pairwise orthogonal relations of states \{1, 2, 3, 5\} is (2, 4, 0) and that of states \{2, 3, 4, 5\} is (4, 2, 0). Thus both the set of \{1, 2, 3, 5\} and the set of states \{2, 3, 4, 5\}, regardless of which cases the two sets belong to, can be seen as a set of bipartite OPSs that are pairwise orthogonal when the third subsystem is not considered. By Lemma~\ref{lemma1}, both states \{1, 2, 3, 5\} and states \{2, 3, 4, 5\} can be perfectly distinguished by LOCC. 
  
  (9) In case (10-13), state 3 is orthogonal to state 5 on the third subsystem, which can be used to identify states \{1, 2, 3, 4\} and states \{1, 2, 4, 5\} by the third party. Both the vector of the numbers of pairwise orthogonal relations of states \{1, 2, 3, 4\} and that of states \{1, 2, 4, 5\} are all (3, 3, 0).
  Thus both the set of \{1, 2, 3, 4\} and the set of states \{1, 2, 4, 5\} can be seen as a set of bipartite OPSs that are pairwise orthogonal when the third subsystem is not considered.  By Lemma~\ref{lemma1}, both states \{1, 2, 3, 4\} and states \{1, 2, 4, 5\} can be perfectly distinguished by LOCC. 
   
  Therefore, five tripartite OPSs with the vector of the numbers of pairwise orthogonal relations $(5,\,4,\,1)$ can be perfectly distinguished by LOCC.
  
  6. Category (5,\,3,\,2)
  
	As shown in Fig.~\ref{fig11}, there exist 30 different orthogonality graphs, i.e., (11-1), (11-2), \dots, (11-30), for five tripartite OPSs with the vector of the numbers of pairwise orthogonal relations (5, 3, 2). 

  (1) In cases (11-1), (11-2), (11-3) and (11-4), state 5 is orthogonal to all the other states on the first subsystem. So the first party can distinguish state 5 from all the others. The result is that four states \{1, 2, 3, 4\} are left to be distinguished. For states \{1, 2, 3, 4\} in these cases, the vectors of the numbers of pairwise orthogonal relations are all (1, 3, 2). As is known, four tripartite OPSs with the vector (1, 3, 2) have the same local distinguishability as four tripartite OPSs with the vector (3, 2, 1). By Lemma~\ref{lemma3}, states \{1, 2, 3, 4\} can be perfectly distinguished by LOCC.     
    
  (2) In case (11-5), state 1 is orthogonal to state 4 and state 5 on the third subsystem, which can be used to identify states \{2, 3, 4, 5\} and states \{1, 2, 3\} by the third party. For states \{2, 3, 4, 5\}, the vector of the numbers of pairwise orthogonal relations is (5, 1, 0). Thus, states \{2, 3, 4, 5\} can be seen as a set of bipartite OPSs that are pairwise orthogonal when the third subsystem is not considered. By Lemma~\ref{lemma1}, states \{2, 3, 4, 5\} can be perfectly distinguished by LOCC. For states \{1, 2, 3\}, these OPSs can be perfectly distinguished by LOCC since they are pairwise orthogonal on the second subsystem. 
  
  (3) In cases (11-6), (11-11), (11-17) and (11-23), state 1 is orthogonal to states $\{2,\,3,\,4\}$ on the second subsystem, which can be used to identify states \{2, 3, 4, 5\} and states \{1, 5\} by the second party. In case (11-6), the vector of the numbers of pairwise orthogonal relations of states \{2, 3, 4, 5\} is (5, 0, 1). In cases (11-11), (11-17) and (11-23), the vectors of the numbers of pairwise orthogonal relations of states \{2, 3, 4, 5\} are all (4, 0, 2). Thus, the set of states \{2, 3, 4, 5\}, regardless of which case it belongs to, can be seen as a set of bipartite OPSs that are pairwise orthogonal when the second subsystem is not considered. By Lemma~\ref{lemma1}, states \{2, 3, 4, 5\} can be perfectly distinguished by LOCC. On the other hand, states \{1, 5\} can be perfectly distinguished by LOCC  by the third party in case (11-6) and can be perfectly distinguished by LOCC by the first party in cases (11-11), (11-17) and (11-23).
    
  (4) In case (11-7), state 1 is orthogonal to state 3 and state 5 on the third subsystem, which can be used to identify states \{2, 3, 4, 5\} and states \{1, 2, 4\} by the third party. The vector of the numbers of pairwise orthogonal relations of states \{2, 3, 4, 5\} is (5,\,1,\,0) and that of states \{1, 2, 4\} is (1, 2, 0).  
  Thus, both the set of states \{2, 3, 4, 5\} and the set of states \{1, 2, 4\} can be seen as a set of bipartite OPSs that are pairwise orthogonal when the third subsystem is not considered. By Lemma~\ref{lemma1}, both states \{2, 3, 4, 5\} and states \{1, 2, 4\} can be perfectly distinguished by LOCC.
   
  (5) In cases (11-8), (11-13) and (11-21), state 3 is orthogonal to state 1 and state 2 on the third subsystem, which can be used to identify states \{1, 2, 4, 5\} and states \{3, 4, 5\} by the third party. In cases (11-8) and (11-13), the vectors of the numbers of pairwise orthogonal relations of states \{1, 2, 4, 5\} are (3, 3, 0) and those of states \{3, 4, 5\} are (3, 0, 0). In case (11-21), the vector of the numbers of pairwise orthogonal relations of states \{1, 2, 4, 5\} is (3, 3, 0) and that of states \{3, 4, 5\} is (2, 1, 0). Thus, both the set of states \{1, 2, 4, 5\} and the set of states \{3, 4, 5\}, regardless of which cases the two sets belong to, can be seen as a set of bipartite OPSs that are pairwise orthogonal when the third subsystem is not considered. By Lemma~\ref{lemma1}, both states \{1, 2, 4, 5\} and states \{3, 4, 5\} can be perfectly distinguished by LOCC.
    
  (6) In cases (11-9), (11-22) and (11-28), state 1 is orthogonal to state 2 and state 3 on the third subsystem, which can be used to identify states \{2, 3, 4, 5\} and states \{1, 4, 5\} by the third party. For case (11-9), the vector of the numbers of pairwise orthogonal relations of states \{2, 3, 4, 5\} is (5,\,1,\,0) and that of states \{1, 4, 5\} is (1, 2, 0). For case (11-22), the vector of the numbers of pairwise orthogonal relations of states \{2, 3, 4, 5\} is (4,\,2,\,0) and that of states \{1, 4, 5\} is (1, 2, 0). For case (11-28), the vector of the numbers of pairwise orthogonal relations of states \{2, 3, 4, 5\} is (4,\,2,\,0) and that of states \{1, 4, 5\} is (2, 1, 0). Thus, both the set of states \{2, 3, 4, 5\} and the set of states \{1, 4, 5\}, regardless of which cases the two sets belong to, can be seen as a set of bipartite OPSs that are pairwise orthogonal when the third subsystem is not considered. By Lemma~\ref{lemma1}, both states \{2, 3, 4, 5\} and states \{1, 4, 5\} can be perfectly distinguished by LOCC.
  
  (7) In case (11-14), state 2 is orthogonal to state 1 and state 5 on the third subsystem, which can be used to identify states \{1, 3, 4, 5\} and states \{2, 3, 4\} by the third party. The vector of the numbers of pairwise orthogonal relations of states \{1, 3, 4, 5\} is (4,\,2,\,0) and that of states \{2, 3, 4\} is (2, 1, 0). Thus, both the set of states \{1, 3, 4, 5\} and the set of states \{2, 3, 4\} can be seen as a set of bipartite OPSs that are pairwise orthogonal when the third subsystem is not considered. By Lemma~\ref{lemma1}, both states \{1, 3, 4, 5\} and states \{2, 3, 4\} can be perfectly distinguished by LOCC.
  
  (8) In case (11-15), state 2 is orthogonal to state 1 and state 3 on the third subsystem, which can be used to identify states \{1, 3, 4, 5\} and states \{2, 4, 5\} by the third party. The vector of the numbers of pairwise orthogonal relations of states \{1, 3, 4, 5\} is (4,\,2,\,0) and that of states \{2, 4, 5\} is (2, 1, 0). Thus, both the set of states \{1, 3, 4, 5\} and the set of states \{2, 4, 5\} can be seen as a set of bipartite OPSs that are pairwise orthogonal when the third subsystem is not considered. By Lemma~\ref{lemma1}, both states \{1, 3, 4, 5\} and states \{2, 4, 5\} can be perfectly distinguished by LOCC.
  
  (9) In case (11-16), state 4 is orthogonal to state 1 and state 5 on the third subsystem, which can be used to identify states \{1, 2, 3, 5\} and states \{2, 3, 4\} by the third party. The vector of the numbers of pairwise orthogonal relations of states \{1, 2, 3, 5\} is (3, 3, 0) and that of states \{2, 3, 4\} is (2, 1, 0). Thus, both the set of states \{1, 2, 3, 5\} and the set of states \{2, 3, 4\} can be seen as a set of bipartite OPSs that are pairwise orthogonal when the third subsystem is not considered. By Lemma~\ref{lemma1}, both states \{1, 2, 3, 5\} and states \{2, 3, 4\} can be perfectly distinguished by LOCC.
   
  (10) In cases (11-20) and (11-27), state 1 is orthogonal to state 3 and state 4 on the third subsystem, which can be used to identify states \{2, 3, 4, 5\} and states \{1, 2, 5\} by the third party. For case (11-20), the vector of the numbers of pairwise orthogonal relations of states \{2, 3, 4, 5\} is (4,\,2,\,0) and that of states \{1, 2, 5\} is (2, 1, 0). For case (11-27), the vector of the numbers of pairwise orthogonal relations of states \{2, 3, 4, 5\} is (4,\,2,\,0) and that of states \{1, 2, 5\} is (1, 2, 0).  Thus, the set of states \{2, 3, 4, 5\} and the set of states \{1, 2, 5\}, regardless of which cases the two sets belong to, can be seen as a set of bipartite OPSs that are pairwise orthogonal when the third subsystem is not considered. By Lemma~\ref{lemma1}, both states \{2, 3, 4, 5\} and states \{1, 2, 5\} can be perfectly distinguished by LOCC.
     
  (11) In case (11-25), state 3 is orthogonal to state 1 and state 5 on the third subsystem, which can be used to identify states \{1, 2, 4, 5\} and states \{2, 3, 4\} by the third party. The vector of the numbers of pairwise orthogonal relations of states \{1, 2, 4, 5\} is (3,\,3,\,0) and that of states \{2, 3, 4\} is (3, 0, 0). Thus, the set of states \{1, 2, 4, 5\} can be seen as a set of bipartite OPSs that are pairwise orthogonal when the third subsystem is not considered. By Lemma~\ref{lemma1}, states \{1, 2, 4, 5\} can be perfectly distinguished by LOCC. On the other hand, states \{2, 3, 4\} can be perfectly distinguished by the first party since any two of states \{2, 3, 4\} are orthogonal on the first subsystem.
    
  (12) In case (11-29), state 5 is orthogonal to state 3 and state 4 on the third subsystem, which can be used to identify states \{1, 2, 3, 4\} and states \{1, 2, 5\} by the third party. The vector of the numbers of pairwise orthogonal relations of states \{1, 2, 3, 4\} is (3,\,3,\,0) and that of states \{1, 2, 5\} is (2, 1, 0). Thus, the set of states \{1, 2, 3, 4\} and the set of states \{1, 2, 5\} can be seen as a set of bipartite OPSs that are pairwise orthogonal when the third subsystem is not considered. By Lemma~\ref{lemma1}, both states \{1, 2, 3, 4\} and states \{1, 2, 5\} can be perfectly distinguished by LOCC.
  
  (13) In cases (11-10) and (11-12), state 4 is orthogonal to states \{ 2,\, 3,\, 5\} and state 5 is orthogonal to states \{1,\, 3,\, 4\} on the first subsystem. We assume that the first subsystems of state 4 and state 5 are $|\alpha\rangle$ and $|\alpha^{\perp }\rangle$, respectively, where $|\alpha\rangle$, $|\alpha^{\perp }\rangle$ are all normalized, and $\langle\alpha|\alpha^{\perp}\rangle=0$. Suppose that the first party, say Alice, performs a measurement with the measurement operators  $|\alpha\rangle \langle\alpha |$, $|\alpha^{\perp} \rangle \langle\alpha^{\perp } |$ and $I-|\alpha\rangle \langle\alpha |- |\alpha^{\perp} \rangle \langle\alpha^{\perp } |$. 
  
  \textcircled{1} If Alice's measurement outcome corresponds to $|\alpha\rangle \langle\alpha |$, the measured state must be state 1 or state 4. States 1 and 4 can be exactly identified by the third party since state 1 and state 4 are orthogonal on the third subsystem. 
  
  \textcircled{2} If Alice's measurement outcome corresponds to $|\alpha^{\perp} \rangle \langle\alpha^{\perp }|$, the measured state must be state 2 and state 5. State 2 and state 5 can be exactly identified by the second party for case (11-12) and by the third party for case (11-10) since state 2 and state 5 are orthogonal on the second subsystem for case (11-12) and on the third subsystem for case (11-10). 
  
  \textcircled{3} If Alice's measurement outcome corresponds to the operator $I-|\alpha\rangle \langle\alpha |- |\alpha^{\perp} \rangle \langle\alpha^{\perp } |$, the measured state must be state 1, 2 or 3. It can be exactly identified since states 1, 2 and 3 are mutually orthogonal by Lemma~\ref{lemma2}.
  
  (14) In cases (11-18) and (11-19), state 5 is orthogonal to states 1, 2 and 3 on the first subsystem. We assume that the first subsystem of state 5 is  $|\alpha\rangle$, where $|\alpha\rangle$ is normalized. Suppose that the first party, say Alice, performs a measurement with the measurement operators  $|\alpha\rangle \langle\alpha |$ and $I-|\alpha\rangle \langle\alpha |$.
  
  \textcircled{1} If Alice's measurement outcome corresponds to the measurement operator $|\alpha\rangle \langle\alpha |$, the measured state must be state 4 or state 5. State 4 and state 5 can be exactly identified by the second party since they are orthogonal on the second subsystem for case (11-18) and can be exactly identified by the third party since they are orthogonal on the third subsystem for case (11-19). 
   
  \textcircled{2} If Alice's measurement outcome corresponds to the measurement operator $I-|\alpha\rangle \langle\alpha |$, the measured state must be one of states \{1, 2, 3, 4\}. Note that we need to verify that state 4 remains orthogonal to state 2 and state 3 after Alice's measurement. We assume that the first subsystems of states 2, 3 and 4 are $|\beta \rangle$,  $|\gamma  \rangle$ and $| \delta\rangle$, respectively, where $|\beta\rangle$, $|\gamma\rangle$ and $|\delta\rangle$ are all normalized, $\langle\beta|\alpha\rangle$=0, $\langle\gamma|\alpha\rangle$=0, $\langle\beta|\delta\rangle$=0, $\langle\gamma|\delta\rangle$=0. The post-measurement states of these first subsystems are $(I-|\alpha\rangle \langle\alpha |)|\beta \rangle$, $(I-|\alpha\rangle \langle\alpha |)|\gamma\rangle$ and $(I-|\alpha\rangle\langle\alpha |)|\delta \rangle$, respectively. The inner product of the first subsystems of state 2 and state 4 after the measurement is $\langle\beta|(I-|\alpha\rangle \langle\alpha|)^{\dagger}(I-|\alpha\rangle\langle\alpha|)|\delta\rangle$=0.  This indicates that state 2 and state 4 remain orthogonal after Alice's measurement. Similarly, state 3 and state 4 remain orthogonal as well after Alice's measurement. Thus states \{1, 2, 3, 4\} are pairwise orthogonal after Alice's measurement. For case (11-18), the vector of the numbers of pairwise orthogonal relations of states \{1, 2, 3, 4\} is (2, 2, 2). From graph (11-18), we know that graph (11-18) is the same as graph (1-1) when the blue edges and the black edges swap colors. It should be note that this swap does not changed the local distinguishability of these four OPSs. By Lemma~\ref{lemma4}, states \{1, 2, 3, 4\} can be perfectly distinguished by LOCC. For case (11-19), the vector of the numbers of pairwise orthogonal relations of states \{1, 2, 3, 4\} is (2, 3, 1). Note that four tripartite OPSs with the vector (2, 3, 1) have the same local distinguishability as four tripartite OPSs with the vector (3, 2, 1). By Lemma~\ref{lemma3}, states \{1, 2, 3, 4\} can be perfectly distinguished by LOCC.

  (15) In case (11-24), state 4 is orthogonal to states \{2, 3, 5\} on the first subsystem. We assume that the first subsystem of state 4 is $|\alpha\rangle$, where $|\alpha\rangle$ is normalized. Suppose that the first party, say Alice, performs a measurement with the measurement operators  $|\alpha\rangle \langle\alpha |$ and $I-|\alpha\rangle \langle\alpha |$ on the first side.

  \textcircled{1} If Alice's measurement outcome corresponds to the measurement operator $|\alpha\rangle \langle\alpha |$, the measured state must be state 1 or state 4. State 1 and state 4 can be exactly identified by the third party since they are orthogonal on the third subsystem.
  
 \textcircled{2} If Alice's measurement outcome corresponds to the measurement operator $I-|\alpha\rangle \langle\alpha |$, the measured state must be one of states \{1, 2, 3, 5\}. Note that we need to verify that state 1 remains orthogonal to state 5 and state 2 remains orthogonal to state 3 on the first subsystem after Alice's measurement. We assume that the first subsystems of states 1, 2, 3 and 5 are $|\beta\rangle$, $|\gamma\rangle$, $|\delta\rangle $ and $|\epsilon\rangle$, respectively, where $|\beta\rangle$, $|\gamma\rangle$, $|\delta\rangle $ and $|\epsilon\rangle$ are all normalized, $\langle\alpha|\epsilon\rangle=0$, $\langle\beta|\epsilon\rangle=0$, $\langle\alpha|\gamma\rangle=0$, $\langle\alpha|\delta\rangle=0$, $\langle\gamma|\delta\rangle=0$. The post-measurement states of these subsystems are $(I-|\alpha\rangle \langle\alpha |)|\beta \rangle$, $(I-|\alpha\rangle \langle\alpha |)|\gamma \rangle$, $(I-|\alpha\rangle \langle\alpha|)|\delta\rangle$ and $(I-|\alpha\rangle \langle\alpha |)|\epsilon \rangle$, respectively. The inner product of the first subsystems of state 1 and state 5 after the measurement is $\langle \beta |(I-|\alpha\rangle \langle\alpha |)^{\dagger }(I-|\alpha\rangle \langle\alpha |)|\epsilon    \rangle $=0.  This indicates that state 1 and state 5 remain orthogonal on the first subsystem after Alice's measurement. Similarly, state 2 and state 3 remain orthogonal on the first subsystem. Thus states \{1, 2, 3, 5\} are pairwise orthogonal after Alice's measurement. The vector of the numbers of pairwise orthogonal relations of states \{1, 2, 3, 5\} is (2, 3, 1). By Lemma~\ref{lemma3}, states \{1, 2, 3, 5\} can be perfectly distinguished by LOCC.
 
  7. Category (4,\,4,\,2)
 
	As shown in Fig.~\ref{fig12}, there exist 20 different orthogonality graphs, i.e., (12-1), (12-2), \dots, (12-20), for five tripartite OPSs with the vector of the numbers of pairwise orthogonal relations (4,\,4,\,2).
		
	(1) In cases (12-1), (12-2), (12-3), (12-4), (12-5), (12-6), (12-7), (12-8), (12-9), (12-10), (12-11) and (12-12), state 5 is orthogonal to states 3 and 4 on the third subsystem, which can be used to identify states \{1, 2, 3, 4\} and states \{1, 2, 5\} by the third party. States \{1, 2, 5\} can be perfectly distinguished by Lemma~\ref{lemma2}.

 \textcircled{1} For cases (12-1), (12-2), (12-4) and (12-7), the vectors of the numbers of pairwise orthogonal relations of states \{1, 2, 3, 4\} are all (2, 4, 0). By Lemma~\ref{lemma1}, these four states can be perfectly distinguished by LOCC when the third subsystem is not considered. 

\textcircled{2} For cases (12-3), (12-5), (12-6), (12-8), (12-9) and (12-10), the vectors of the numbers of pairwise orthogonal relations of states \{1, 2, 3, 4\} are all (3, 3, 0). By Lemma~\ref{lemma1}, these four states can be perfectly distinguished by LOCC when the third subsystem is not considered. 

\textcircled{3} For cases (12-11) and (12-12), the vectors of the numbers of pairwise orthogonal relations of states \{1, 2, 3, 4\} are all (4, 2, 0). By Lemma~\ref{lemma1}, states \{1, 2, 3, 4\} can be perfectly distinguished by LOCC when the third subsystem is not considered.

	(2) In case (12-13), state 1 is orthogonal to states 2, 3 and 4 on the second subsystem. We assume that the second subsystem of state 1 is  $|\alpha\rangle$, where $|\alpha\rangle$ is normalized. Suppose that the second party, say Bob, performs a measurement with the measurement operators $|\alpha\rangle \langle\alpha |$ and $I-|\alpha\rangle \langle\alpha |$. 

\textcircled{1} If Bob's measurement outcome corresponds to the operator $|\alpha\rangle \langle\alpha |$, the measured state must be state 1 or state 5. State 1 and state 5 can be exactly identified by the first party since they are orthogonal on the first subsystem. 

\textcircled{2} If Bob's measurement outcome corresponds to the operator $I-|\alpha\rangle \langle\alpha |$, the measured state must be one of states \{2, 3, 4, 5\}. Note that state 2 and state 3 remain orthogonal on the second system after Bob's measurement. The vector of the numbers of pairwise orthogonal relations of states \{2, 3, 4, 5\} is (3, 1, 2). By Lemma~\ref{lemma3}, states \{2, 3, 4, 5\} can be perfectly distinguished by LOCC.
	
	(3) In case (12-14), state 2 is orthogonal to states 1, 3 and 4 on the second subsystem. We assume that the second subsystem of state 2 is $|\alpha\rangle$, where $|\alpha\rangle$ is normalized. Suppose that the second party, say Bob, performs a measurement with the measurement operators $|\alpha\rangle \langle\alpha |$ and $I-|\alpha\rangle \langle\alpha |$. 

\textcircled{1} If Bob's measurement outcome corresponds to  $|\alpha\rangle \langle\alpha |$, the measured state must be state 2 or state 5. State 2 and state 5 can be exactly identified by the third party since they are orthogonal on the third subsystem. 

\textcircled{2} If Bob's measurement outcome corresponds to $I-|\alpha\rangle \langle\alpha|$, the measured state must be one of states \{1, 3, 4, 5\}. Note that state 1 remains orthogonal to state 3 on the second subsystem after Bob's measurement. The vector of the numbers of pairwise orthogonal relations is (4, 1, 1).
By Lemma~\ref{lemma3}, states \{1, 3, 4, 5\} can be perfectly distinguished by LOCC. 
	
	(4) In case (12-16), state 1 is orthogonal to all the other states on the second subsystem. Thus the second party can identify state 1 from all the others. The result is that four states \{2, 3, 4, 5\} are left to be distinguished. The vector of the numbers of pairwise orthogonal relations of states \{2, 3, 4, 5\} is (4, 0, 2).  By Lemma~\ref{lemma1}, states \{2, 3, 4, 5\} can be perfectly distinguished by LOCC when the second subsystem is not considered.
	
	(5) In case (12-17), state 1 is orthogonal to states 2, 3 and 4 on the second subsystem.  We assume that the second subsystem of state 1 is  $|\alpha\rangle$, where $|\alpha\rangle$ is normalized. Suppose that the second party, say Bob, performs a measurement with the measurement operators  $|\alpha\rangle \langle\alpha |$ and $I-|\alpha\rangle \langle\alpha |$. 

\textcircled{1} If Bob's measurement outcome corresponds to the operator $|\alpha\rangle \langle\alpha |$, the measured state must be state 1 or state 5. State 1 and state 5 can be exactly identified by the first party since they are orthogonal on the first subsystem. 

\textcircled{2} If Bob's measurement outcome corresponds to the operator $I-|\alpha\rangle \langle\alpha |$, the measured state must be one of states \{2, 3, 4, 5\}. Note that state 3 and state 5 remain orthogonal on the second subsystem after Bob's measurement. Thus states \{2, 3, 4, 5\} are pairwise orthogonal. The vector of the numbers of pairwise orthogonal relations of states \{2, 3, 4, 5\} is (3, 1, 2). By Lemma~\ref{lemma3}, states \{2, 3, 4, 5\} can be perfectly distinguished by LOCC.
	
	(6) In case (12-18), state 5 is orthogonal to state 1 and state 3 on the first subsystem.  We assume that the first subsystem of state 5 is $|\alpha\rangle$, where $|\alpha\rangle$ is normalized. Suppose that the first party, say Alice, performs a measurement with the measurement operators  $|\alpha\rangle \langle\alpha |$ and $I-|\alpha\rangle \langle\alpha |$. 

\textcircled{1} If Alice's measurement outcome corresponds to the operator $|\alpha\rangle \langle\alpha |$, the measured state must be state 2, state 4 or state 5. States \{2, 4, 5\} can be locally distinguished by Lemma~\ref{lemma2} since they still remain pairwise orthogonal after Alice's measurement. 

\textcircled{2} If Alice's measurement outcome corresponds to the operator $I-|\alpha\rangle \langle\alpha |$, the measured state must be one of states $\{1,\,2,\,3,\,4\}$. Note that state 2 is orthogonal to state 3 and state 1 is orthogonal to state 4 on the first subsystem after Alice's measurement. Thus, states \{1, 2, 3, 4\} are still pairwise orthogonal after Alice's measurement. The vector of the numbers of pairwise orthogonal relations of states \{1, 2, 3, 4\} is (2, 3, 1). By Lemma~\ref{lemma3}, states \{1, 2, 3, 4\} can be perfectly distinguished by LOCC. 

	(7) In case (12-19), state 2 is orthogonal to states 1, 3 and 4 on the second subsystem. We assume that the second subsystem of state 2 is  $|\alpha\rangle$, where $|\alpha\rangle$ is normalized. Suppose that the second party, say Bob, performs a measurement with the measurement operators $|\alpha\rangle \langle\alpha |$ and $I-|\alpha\rangle \langle\alpha |$. 

\textcircled{1} If Bob's measurement outcome corresponds to $|\alpha\rangle\langle\alpha |$, the measured state must be state 2 or state 5. State 2 and state 5 can be exactly identified by the third party since they are orthogonal on the third subsystem. 

\textcircled{2} If Bob's measurement outcome corresponds to $I-|\alpha\rangle \langle\alpha |$, the measured state must be one of states $\{1,\,3,\,4,\,5\}$. Note that state 1 and state 5 remain orthogonal on the second subsystem after Bob's measurement. Thus states $\{1,\,3,\,4,\,5\}$ remain pairwise  orthogonal after Bob's measurement. The vector of the numbers of pairwise orthogonal relations of states \{1, 3, 4, 5\} is (4, 1, 1). By Lemma~\ref{lemma3}, states \{1, 3, 4, 5\} can be perfectly distinguished by LOCC.

  8. Category (4,\,3,\,3)
  
	As shown in Fig.~\ref{fig13}, there exist 27 different orthogonality graphs, i.e., (13-1), (13-2), \dots, (13-27), for five tripartite OPSs with the vector of the numbers of pairwise orthogonal relations $(4,\,3,\,3)$.

	(1) In cases (13-1) and (13-2), state 5 is orthogonal to all the other states on the first subsystem. Thus the first party can distinguish state 5 from all the others. The result is that four states \{1, 2, 3, 4\} are left to be distinguished. The vector of the numbers of pairwise orthogonal relations of states \{1, 2, 3, 4\} is (0, 3, 3). By Lemma~\ref{lemma1}, states \{1, 2, 3, 4\} can be perfectly distinguished by LOCC when the first subsystem is not considered.    
	
	(2) In cases (13-3), (13-5) and (13-8), state 5 is orthogonal to states 2, 3 and 4 on the first subsystem. We assume that the first subsystem of state 5 is  $|\alpha\rangle$, where $|\alpha\rangle$ is normalized. Suppose that the first party, say Alice, performs a measurement with the measurement operators  $|\alpha\rangle \langle\alpha |$ and $I-|\alpha\rangle \langle\alpha |$. 

\textcircled{1} If Alice's measurement outcome corresponds to the operator $|\alpha\rangle \langle\alpha |$, the measured state must be state 1 or state 5. For cases (13-3) and (13-5), state 1 and state 5 can be exactly identified by the third party since they are orthogonal on the third subsystem. For case (13-8), the two states can also be exactly identified by the second party due to their orthogonality on the second subsystem.

 \textcircled{2} If Alice's measurement outcome corresponds to $I-|\alpha\rangle \langle\alpha |$, the measured state must be one of states $\{1,\,2,\,3,\,4\}$. Note that state 3 and state 4 remain orthogonal after Alice's measurement. This means that states \{1, 2, 3, 4\} remain pairwise orthogonal after Alice's measurement. For cases (13-3) and (13-5), the vectors of the numbers of pairwise orthogonal relations of states \{1, 2, 3, 4\} are all (1, 3, 2). By Lemma~\ref{lemma3}, states \{1, 2, 3, 4\} can be perfectly distinguished by LOCC. For case (13-8), the vector of the numbers of pairwise orthogonal relations of states \{1, 2, 3, 4\} is (1, 2, 3). By Lemma~\ref{lemma3}, states \{1, 2, 3, 4\} can be perfectly distinguished by LOCC.
	
	(3) In cases (13-4), (13-10), (13-16), (13-20) and (13-25), state 1 is orthogonal to states \{ 2, 3, 4\} on the second subsystem, which can be used to identify states \{2, 3, 4, 5\} and states \{1, 5\} by the second party. 

\textcircled{1} For cases (13-4) and (13-16), the vectors of the numbers of pairwise orthogonal relations of states \{2, 3, 4, 5\} are all (4, 0, 2). Thus states \{2, 3, 4, 5\} can be seen as a set of bipartite OPSs when the second subsystem is not considered. By Lemma~\ref{lemma1}, states \{2, 3, 4, 5\} can be perfectly distinguished by LOCC. For case (13-10), (13-20) and (13-25), the vectors of the numbers of pairwise orthogonal relations of states \{2, 3, 4, 5\} are all (3, 0, 3). Thus states \{2, 3, 4, 5\} can be seen as a set of bipartite OPSs when the second subsystem is not considered. By Lemma~\ref{lemma1}, states \{2, 3, 4, 5\} can be perfectly distinguished by LOCC. 

\textcircled{2} For states \{1, 5\}, they can be perfectly distinguished by LOCC since they are orthogonal in each of cases (13-4), (13-10), (13-16), (13-20) and (13-25). 
	
	(4) In case (13-6), state 1 is orthogonal to states 2, 3 and 5 on the second subsystem, which can be used to identify states \{2, 3, 4, 5\} and states \{1, 4\} by the second party. The vector of the numbers of pairwise orthogonal relations of states \{2, 3, 4, 5\} is (4, 0, 2). By Lemma~\ref{lemma1}, states \{2, 3, 4, 5\} can be perfectly distinguished by LOCC. For states \{1, 4\}, they can be perfectly distinguished by the third party since they are orthogonal on the third subsystem.  
	
	(5) In case (13-7), state 2 is orthogonal to states 1, 3 and 4 on the second subsystem, which can be used to identify states \{1, 3, 4, 5\} and states \{2, 5\} by the second party. The vector of the numbers of pairwise orthogonal relations of states \{1, 3, 4, 5\} is (3, 0, 3). By Lemma~\ref{lemma1}, states \{1, 3, 4, 5\} can be perfectly distinguished by LOCC. States \{2, 5\}  can be perfectly distinguished by the first party since they are orthogonal on the first subsystem.   
	
	(6) In cases (13-9), (13-11), (13-12), (13-13) and (13-14), state 5 is orthogonal to states 1, 2 and 4 on the first subsystem. We assume that the first subsystem of state 5 is  $|\alpha\rangle$, where $|\alpha\rangle$ is normalized. Suppose that the first party, say Alice, performs a measurement with the measurement operators  $|\alpha\rangle \langle\alpha |$ and $I-|\alpha\rangle \langle\alpha |$. 

   \textcircled{1} If Alice's measurement outcome corresponds to $|\alpha\rangle\langle\alpha |$, the measured state must be state 3 or state 5. For cases (13-9), (13-11) and (13-13), state 3 and state 5 can be exactly identified by the third party since they are orthogonal on the third subsystem. For cases (13-12) and (13-14), state 3 and state 5 can be exactly identified by the second party since they are orthogonal on the second subsystem. 

   \textcircled{2} If Alice's measurement outcome corresponds to $I-|\alpha\rangle \langle\alpha |$, the measured state must be one of states \{1, 2, 3, 4\}. Note that state 3 remains orthogonal to state 4 on the first subsystem after Alice's measurement. For cases (13-9), (13-11) and (13-13), the vectors of the numbers of pairwise orthogonal relations of states \{1, 2, 3, 4\} are all (1, 3, 2). For cases (13-12) and (13-14), the vectors of the numbers of pairwise orthogonal relations of states \{1, 2, 3, 4\} are all (1, 2, 3). By Lemma~\ref{lemma3}, states \{1, 2, 3, 4\} in each of cases (13-9), (13-11), (13-12), (13-13) and (13-14) can be perfectly distinguished by LOCC. 

	(7) In cases (13-15), (13-17) and (13-18), state 5 is orthogonal to states 2 and 3 on the first subsystem. We assume that the first subsystem of state 5 is 
   $|\alpha\rangle$, where $|\alpha\rangle$ is normalized. Suppose that the first party, say Alice, performs a measurement with the measurement operators  $|\alpha\rangle \langle\alpha |$ and $I-|\alpha\rangle \langle\alpha |$. 

   \textcircled{1} If Alice's measurement outcome corresponds to the operator $|\alpha\rangle \langle\alpha |$, the measured state be state 1, 4 or 5. States \{1, 4, 5\} can be locally distinguished by Lemma~\ref{lemma2} since they are still pairwise orthogonal. 

   \textcircled{2} If Alice's measurement outcome corresponds to the operator  $I-|\alpha\rangle \langle\alpha |$, the measured state must be one of states $\{1,\,2,\,3,\,4\}$. Note that state 4 remains orthogonal to state 2 and state 3 on the first subsystem under this measurement outcome of Alice. Thus states \{1, 2, 3, 4\} are still pairwise orthogonal. For cases (13-15) and (13-18), the vectors of the numbers of pairwise orthogonal relations of states \{1, 2, 3, 4\} are all (2, 3, 1). By Lemma~\ref{lemma3}, states \{1, 2, 3, 4\} in case (13-15) or (13-18)  can be perfectly distinguished by LOCC. For case (13-17), the vector of the numbers of pairwise orthogonal relations of states \{1, 2, 3, 4\} is (2, 2, 2) and the graph of states \{1, 2, 3, 4\} corresponds to graph (1-1). By Lemma~\ref{lemma4}, states \{1, 2, 3, 4\} in  case (13-17) can be perfectly distinguished by LOCC.

	(8) In case (13-19), state 1 is orthogonal to state 2 and state 3 on the second subsystem. We assume that the second subsystem of state 1 is $|\alpha\rangle$, where $|\alpha\rangle$ is normalized. Suppose that the second party, say Bob, performs a measurement with the measurement operators  $|\alpha\rangle \langle\alpha |$ and $I-|\alpha\rangle \langle\alpha |$.
	
	\textcircled{1} If Bob's measurement outcome corresponds to the operator $|\alpha\rangle \langle\alpha |$, the measured state must be state 1, 4 or 5. States \{1, 4, 5\} can be locally distinguished by Lemma~\ref{lemma2} since they are still pairwise orthogonal after Bob's measurement.
	
    \textcircled{2} If Bob's measurement outcome corresponds to $I-|\alpha\rangle \langle\alpha |$, the measured state must be one of states $\{2,\,3,\,4,\,5\}$. 
  The vector of the numbers of pairwise orthogonal relations of states $\{2,\,3,\,4,\,5\}$ is (3, 1, 2). By Lemma~\ref{lemma3}, states $\{2,\,3,\,4,\,5\}$  can be perfectly distinguished by LOCC.

	(9) In cases (13-21) and (13-23) , state 2 is orthogonal to state 1 and state 5 on the second subsystem. We assume that the second subsystem of state 2 is  
  $|\alpha\rangle$, where $|\alpha\rangle$ is normalized. Suppose that the second party, say Bob, performs a measurement with the measurement operators  $|\alpha\rangle \langle\alpha |$ and $I-|\alpha\rangle \langle\alpha |$.
	
	\textcircled{1} If Bob's measurement outcome corresponds to $|\alpha\rangle \langle\alpha |$, the measured state must be state 2, 3 or 4. States \{2, 3, 4\} can 
   be locally distinguished by Lemma~\ref{lemma2} since they are pairwise orthogonal.
	
   \textcircled{2} If Bob's measurement outcome corresponds to  $I-|\alpha\rangle \langle\alpha |$, the measured state must be one of states $\{1,\,3,\,4,\,5\}$. It is easy to know that state 1 is still orthogonal to state 3 on the second subsystem for case (13-21) and state 1 is still orthogonal to state 4 on the second subsystem for case (13-23). For cases (13-21) and (13-23), the vectors of the numbers of pairwise orthogonal relations of states $\{1,\,3,\,4,\,5\}$ are all (3, 1, 2). By Lemma~\ref{lemma3}, states $\{1,\,3,\,4,\,5\}$ can be perfectly distinguished by LOCC since they are still pairwise orthogonal after Bob's measurement.
	
   (10) In case (13-24), state 1 is orthogonal to state 2 and state 4 on the second subsystem.  We assume that the second subsystem of state 1 is  $|\alpha\rangle$, where $|\alpha\rangle$ is normalized. Suppose that the second party, say Bob, performs a measurement with the measurement operators  $|\alpha\rangle \langle\alpha |$ and $I-|\alpha\rangle \langle\alpha|$.
	
	\textcircled{1} If Bob's measurement outcome corresponds to $|\alpha\rangle \langle\alpha |$, the measured state must be state 1, 3 or 5. By Lemma~\ref{lemma2}, states 
    \{1, 3, 5\} can be perfectly distinguished by LOCC.
	
    \textcircled{2} If Bob's measurement outcome corresponds to  $I-|\alpha\rangle \langle\alpha |$, the measured state must be one of states $\{2,\,3,\,4,\,5\}$. Note that state 4 remains orthogonal to state 5 on the second party. Thus states $\{2,\,3,\,4,\,5\}$ are still pairwise orthogonal after Bob's measurement. The vector of the numbers of pairwise orthogonal relations of states $\{2,\,3,\,4,\,5\}$ is (3, 1, 2). By Lemma~\ref{lemma3}, states  $\{2,\,3,\,4,\,5\}$ can be perfectly distinguished by LOCC. 

	(11) In case (13-26), state 1 is orthogonal to state 2 and state 3 on the second subsystem. We assume that the second subsystem of state 1 is  $|\alpha\rangle$, 
    where $|\alpha\rangle$ is normalized. Suppose that the second party, say Bob, performs a measurement with the measurement operators $|\alpha\rangle \langle\alpha|$ and $I-|\alpha\rangle \langle\alpha|$.
	
    \textcircled{1} If Bob's measurement outcome corresponds to the operator $|\alpha\rangle \langle\alpha|$, the measured state must be one of states \{1, 4, 5\}. States \{1, 4, 5\} can be locally distinguished by Lemma~\ref{lemma2} since they are still orthogonal after Bob's measurement. 
	
	\textcircled{2} If Bob's measurement outcome corresponds to the operator $I-|\alpha\rangle \langle\alpha|$, the measured state must be one of states 
    \{2, 3, 4, 5\}. Note that state 2 and state 5 remain orthogonal on the second subsystem. Thus states \{2, 3, 4, 5\} are still pairwise orthogonal. The vector of the numbers of pairwise orthogonal relations of states $\{2,\,3,\,4,\,5\}$ is (3, 1, 2). By Lemma~\ref{lemma3}, states  $\{2,\,3,\,4,\,5\}$ can be perfectly distinguished by LOCC. 
    
    In summary, five tripartite OPSs, any two of which are orthogonal only on one subsystem, can be perfectly distinguished by LOCC, except those with graph (11-30), (9-8), (11-26), (12-15), (12-20), (13-22) or (13-27). This completes the proof.
    \end{proof}

\nocite{*}
\bibliography{apssamp}
\end{document}